\documentclass[11pt,a4,english]{article}
\usepackage[latin9]{inputenc}
\usepackage{geometry}
\usepackage{verbatim}
\usepackage{amsmath}
\usepackage{amssymb}
\usepackage{mathtools}
\usepackage{graphicx}
\usepackage{longtable}
\usepackage{setspace}
\usepackage{upgreek}
\usepackage{algorithm}
\usepackage{algorithmic}

\usepackage[round,sort,authoryear]{natbib}
\usepackage{amsfonts}
\usepackage{lscape}
\usepackage[colorlinks=true,linkcolor=blue, citecolor=blue, urlcolor = blue]{hyperref}

\makeatletter
\usepackage{graphicx}
\makeatother
\usepackage{babel}
\usepackage{amsfonts}
\usepackage{amsthm}
\usepackage[toc,page]{appendix}
\usepackage{booktabs,caption}
\usepackage{multirow}
\usepackage[labelfont=bf,textfont=md]{caption}
\usepackage[labelsep=space]{caption}
\usepackage{hyperref}
\usepackage{lmodern}
\usepackage[T1]{fontenc}

\usepackage{mathtools}

\makeatletter
\newcommand{\nextverbatimspread}[1]{%
  \def\verbatim@font{%
    \linespread{#1}\normalfont\ttfamily% Updated definition
    \gdef\verbatim@font{\normalfont\ttfamily}}% Revert to old definition
}
\makeatother

\usepackage{natbib}
\usepackage{bibentry}
\nobibliography*

\usepackage{titlesec}
\titleformat*{\section}{\large \bfseries}
\titleformat*{\subsection}{\normalsize \bfseries}
\titleformat*{\subsubsection}{\normalsize \bfseries}

\numberwithin{equation}{section}

\newtheorem{condition}{Condition}

\theoremstyle{definition}
\newtheorem{theorem}{Theorem}
\newtheorem{definition}{Definition}
\newtheorem{assumption}{Assumption}
\newtheorem{lemma}{Lemma}
\newtheorem{example}{Example}
\newtheorem{proposition}{Proposition}
\newtheorem{remark}{Remark}

\usepackage{amsmath}
\newcommand\norm[1]{\left\lVert#1\right\rVert}

\begin{document}
\pagenumbering{roman}

%A Graph Topology Measure of a Time Series Regression-based \\ Risk Matrix with Tail Estimates

\title{ {\Large \textbf{Statistical Estimation for Covariance Structures with Tail Estimates using Nodewise Quantile Predictive Regression Models}\thanks{The current study was initiated during my PhD studies at the University of Southampton. The first manuscript was titled: "\textit{Optimal Portfolio Allocation Using Financial Networks}". The paper was presented during the PhD Workshop in March 2019 and in October 2019 at the Department of Economics, University of Southampton. A revised manuscript was titled: "\textit{Optimal Portfolio Choice and Stock Centrality under Tail Events}" which can be found on \href{https://arxiv.org/abs/2112.12031}{ArXiv:2112.12031}. \textbf{Article history:} December 2021, March 2022, May 2023.  
\\

I wish to thank my main PhD supervisor Prof. Jose Olmo for guidance and continuous encouragement throughout the PhD programme as well as Prof. Peter W. Smith for helpful discussion and for stimulating the further development of the current study in the direction of node exclusion and feature selection. In addition, I wish to thank Zudi Lu, Jean-Yves Pitarakis, Tassos Magdalinos, Hector Calvo Pardo, Zacharias  Maniadis, Ruben Sanchez Garcia and Tullio Mancini as well as Departmental Seminar Series speakers at the University of Southampton: Luis E. Candelaria, Juan Carlos Escanciano, Marcelo C. Medeiros, Sebastian  Engelke, Markus Pelger and Loriano Mancini for helpful conversations. 
\\

Lecturer in Economics, University of Exeter Business School, Exeter EX4 4PU, United Kingdom. Email: \textcolor{blue}{c.katsouris@exeter.ac.uk}
} 
}
}

%\href{https://arxiv.org/abs/2112.12031}{(arXiv:2112.12031)}

%Ph.D. Graduate, Department of Economics, University of Southampton, Southampton, SO17 1BJ, UK.

\author{\textbf{Christis Katsouris}\\ \textit{University of Exeter} }
%}
%\\
%\small \textit{Department of Economics, University of Southampton} 
%\\ 
%\small Southampton SO17 1BJ, United Kingdom
%\\
%\small
%\texttt{C.Katsouris@soton.ac.uk}
%\\
%\small \textit{Department of Economics, University of Exeter}
%\\
%\small Exeter EX4 4PU, United Kingdom
%\small First Draft: 19 March 2019
%}

%\href{https://sites.google.com/view/christiskatsouris/home}{Website}.

%\small{Department of Economics, University of Southampton}

\date{July 24, 2023}

\maketitle

%proposes a novel graph topology measure for two multivariate random variables $\boldsymbol{Y}_t$ and $\boldsymbol{X}_{t-1}$ such that $\boldsymbol{Y} = \left( Y_{1t},... Y_{kt} \right)^{\prime}$ and $\boldsymbol{X} = \left( X_{1t-1},... X_{pt-1} \right)^{\prime}$ are time series observations with one time-lag difference, which are employed to construct a novel regression-based risk matrix. The tail risk matrix operates under the assumption of stationarity and is constructed with nodewise quantile predictive regressions. 

% based on nodewise quantile predictive regression models. 

% based on our regression-based risk matrix, we 
 
\begin{abstract}
\vspace*{-0.3 em}
This paper considers the specification of covariance structures with tail estimates. We focus on two aspects: (i) the estimation of the $\mathsf{VaR-\Delta CoVaR}$ risk matrix in the case of larger number of time series observations than assets in a portfolio using quantile predictive regression models without assuming the presence of nonstationary regressors and; (ii) the construction of a novel variable selection algorithm, so-called \textit{Feature Ordering by Centrality Exclusion} (FOCE), which is based on an assumption-lean regression framework, has no tuning parameters and is proved to be consistent under general sparsity assumptions. We illustrate the usefulness of our proposed methodology with numerical studies of real and simulated datasets when modelling systemic risk in a network. 
\\


\vspace{1.2pt}

\end{abstract}

%In this paper motivated by the portfolio risk optimization in financial networks survey study, we examine the relationship between portfolio risk and eigenvector centrality based on the global minimum variance portfolio. This particular aspect is of special interest when an economic agent considers the optimal portfolio choice within a network topology. 

%In other words, we introduce a novel risk matrix and as a by product we also introduce a novel risk measure. 

%In this paper we propose a novel algorithm for feature selection, we call "Feature Sorting by Centrality Exclusion" (FSCE).

%We consider the eigenvector centrality as a statistical measure that characterizes the degree of connectedness of nodes within a network and show that there is no monotonic relationship between the centrality of an asset and its optimal portfolio allocations as previously thought.     

%and centrality measures 

%, while in the finance literature a recent research question of interest is the optimal portfolio allocation problem in financial networks (see, \cite{peralta2016network} and \cite{olmo2021optimal}). 

%\paragraph{Simultaneous Causation}

%The construction of the proposed risk matrix based on the nodewise quantile predictive regression models can be employed to test for the presence of systemic risk effects. In a sense the particular testing methodology can be considered as testing for a simultaneous causation, a terminology defined by Rosenberg (1998). 

%%-------------------------------------------------------------------------%%
\newpage 

\setcounter{page}{1}
\pagenumbering{arabic}

\section{Introduction.}

Network analysis has seen  in recent decades a growing attention in the statistics, econometrics and quantitative finance literature; with a common aspect of consideration the role of node centrality in model estimation and portfolio optimization problems. Specifically, in the statistics literature related applications include graphical models and node exclusion tests (see, \cite{roverato1998isserlis} and \cite{salgueiro2005power}), while in the econometrics literature recent applications investigate the modelling of interconnectedness and financial spillover effects based on graph structures that capture interactions among economic agents (see, \cite{hong2009granger}, \cite{billio2012econometric},  \cite{diebold2012better}, \cite{hardle2016tenet}, \cite{blasques2016spillover},  \cite{barigozzi2017network}, \cite{barunik2018measuring}). In the finance and risk management literature, a relevant aspect of interest is the optimal portfolio allocation problem using financial networks\footnote{A specific research question of these studies is the relation between stock centrality and portfolio risk. \cite{peralta2016network} examine the minimum variance and mean-variance asset allocation problems by connecting the optimal portfolio weights of each investment strategy to the eigenvector centrality. The authors use the correlation decomposition of the covariance matrix to represent the network topology. Thus, the variance-covariance matrix is decomposed into a quadratic form with the correlation matrix as a quadratic matrix and a diagonal matrix with the variances of the assets. By defining a correlation matrix as a matrix which has $1'$s on its main diagonal and the pairwise correlations on its off diagonal elements then this matrix is considered as an adjacency matrix.} (see, \cite{peralta2016network}, \cite{olmo2021optimal},  \cite{katsouris2021optimal} and \cite{lin2022portfolio}). All aforementioned studies consider suitable statistical dependence measures that are robust to deformations of the network topology, a commonly occurring phenomenon during a financial crisis. Our main research objective is to investigate methodologies for modelling extreme events in graphs combining both tail risk and the induced centrality measures. 

Our first contribution to the literature is the construction of a regression-based covariance-type matrix for modelling tail dependency in graphs using pairwise quantile regression models\footnote{The econometric specification for obtaining the risk measures of VaR and CoVaR are presented in the studies of \cite{adrian2016covar} and \cite{hardle2016tenet} (see also, \cite{white2015var}).}; which to the best our knowledge is a novel estimation approach. Specifically, the proposed risk matrix captures higher-order dependency induced by the tail behaviour of the underline distributions of a fixed graph\footnote{A graph representation of a financial network with the proposed risk matrix is proposed by \cite{katsouris2021optimal} who show that this covariance-type induces a positive-definite quadratic form which permits to consider optimal portfolio allocation problems under tail events. In this paper, the focus is on the proposed feature selection algorithm.}, under the assumption of time series stationarity\footnote{In this paper, we consider that the pairwise quantile predictive regression models are driven only by the predictive regression equation without using a nonstationary specification for regressors. An extension to the case in which the proposed tail-risk matrix depends nonlinearly on the degree of persistence is currently work in progress by the author.}. Although we do not directly model the multivariate extreme dependence as in the study of \cite{engelke2020graphical}, the constructed large covariance-type matrix captures a general form of network interconnectedness through pairwise graph causality. In particular, the conditional tail (extremal) dependency across the nodes of the graph is measured in relation to the risk measures of VaR and CoVaR obtained as tail estimates (see, Section \ref{Section2}).

Our second contribution to the literature, is a novel feature selection algorithm based on the proposed risk matrix, that we call \textit{Feature Ordering by Centrality Exclusion} (FOCE). This node selection algorithm provides an ordering mechanism using the parametrically specified covarariance type matrix that corresponds to a fixed quantile level, denoted with $\boldsymbol{\Gamma}_n ( \uptau )$ for some $\uptau \in (0,1)$, which implies that the covariance structure is approximated with tail estimates obtained from pairwise quantile (predictive) regressions fitted on a sample size of $n$ observations. Recently, \cite{caner2022sharpe} present a framework for nodewise regression when the residuals from a fitted factor model are used and a feasible residual-based nodewise regression is proposed to estimate the precision matrix of errors.

%%-------------------------------------------------------------------------%%
\newpage

In this paper we focus on the introduction of the novel risk matrix as well as in the introduction of the novel graph topology measure. In particular, the proposed risk matrix captures tail dependence and causality in graphs and thus the proposed algorithm of feature ordering is based on this conditional tail dependence measure providing this way a methodology to order characteristic statistics of the graph topology such as centrality measures based on the tail risk matrix. Specifically, the proposed algorithm has an oracle property in the sense that if we plug the true population parameters then our algorithm gives the true feature selection by centrality exclusion of the population variables. Using the proposed algorithm and risk matrix, we aim to model multivariate tail behaviour by considering the tail behaviour of each pair of nodes in the graph. More importantly, it appears to have good performance in simulated and real data sets. Therefore, we focus on the development of the FSCE as well as the proofs of its consistency. The proposed algorithm is not possible to be implemented in low-dimensional models due to convergence and consistency issues based on the proposed stopping rule of the algorithm. We aim to investigate whether our procedure can be estimated with respect to convergence rate $\mathcal{O} \left( n \ \mathsf{log} n \right)$ where $n$ is the sample size.

Usually in the literature the main concern is to show that for $p >> n$ the Sharpe Ratio estimators are consistent in global minimum-variance and mean-variance portfolios. However, in this paper our main concern is to demonstrate the use of the proposed tail-risk matrix in portfolio optimization problems. Specifically, covariance matrices estimated using high dimensional statistical models are used in asset allocation settings as a precision matrix (inverse of the covariance matrix). On the other hand, in high dimensional settings a sparsity condition on the precision matrix of the errors from a statistical model (such as a factor model, a nodewise quantile predictive regression model etc.), can be justified in empirical finance studies. Furthermore, the statistical interpretation of of zero entries on a precision matrix can be attributed into conditional independence.  Then, our proposed methodology focuses on estimating a covariance-type matrix using the tail forecasts rather than based on the estimated residuals from a factor model as in the study of \cite{caner2022sharpe}.

\subsection{Illustrative Examples}

\begin{example}
Consider the conventional covariance structure that corresponds to the errors of the factor model (see, \cite{fan2015risks}, \cite{caner2022sharpe} among others) such that 
\begin{align}
\label{covariance}
\hat{\sigma}_{ij}^o = \frac{1}{T} \sum_{t=1}^T u_{it} u_{jt} \ \ \ \text{and} \ \ \ \hat{\omega}_{ij}^o = \frac{ \hat{\sigma}_{ij}^o  }{ \big( \hat{\sigma}_{ii}^o  \hat{\sigma}_{jj}^o  \big)^{1/2} }    
\end{align}
where the $\sigma^o$ notation stands for oracle, which indicates that these estimators are not feasible because the true model errors are required. Clearly, the covariance structure given by \eqref{covariance}, implies that all related statistical concepts such as deriving probability bounds for the estimation error, deriving consistent estimators of portfolio performance measures (such as sharpe ratios) in a global minimum-variance or mean-variance portfolio, will depend on the unknown error of the underline factor model. On the other hand, the novelty of our proposed framework is that the covariance matrix under consideration corresponds to the true values of the tail risk measures, $\mathsf{VaR}-\mathsf{CoVaR}$, which are estimated as forecasts\footnote{Notice that for the purpose of this paper, we assume that these risk measures have good elicitability properties. } from nodewise quantile predictive regression models rather than their errors as it is the common practice in the literature (see, \cite{callot2021nodewise} and \cite{caner2022sharpe}). 
\end{example}

%%-------------------------------------------------------------------------%%
\newpage

A different stream of literature considers shrinkage methods such as the framework proposed by \cite{ledoit2012nonlinear}, who consider a nonlinear shrinkage estimator in which small eigenvalues of the sample covariance matrix are increased and large eigenvalues are decreased. Although methodologies for shrinkaging the eigenvalues of covariance matrices which are used for asset allocation purposes ensure that the optimization is not ill-conditioned, this aspect is of independent interest since our main contribution for this paper is to study the main properties of this novel tail-risk matrix. A conjecture is that similar to the case when the covariance matrix is estimated from the errors of factor models, a larger number of included factors will result on a slower convergence rate of the estimation error. Similarly, when the tail-based covariance matrix is estimated with nodewise quantile predictive regressions, an increased number of regressors (either stationary or nonstationary) will result to a slower convergence of the estimation error. Notice that with estimation error here we mean the error in estimating the true structure of the covariance matrix. Moreover, the conventional approach in the literature is to separate between the number of units (e.g., $p$ nodes in the network) versus the number of time series observations. In our case we have another relevant parameter which is the number of regressors included in each quantile predictive regression model (similar to the case of factors when considering factor models). In this paper, we focus in the case in which the number of nodes in the network is not necessarily much larger than the number of time series observations, that is, $p < n$ (or $N < T$ in an equivalent notation).

In the conventional factor structure model case, the following theorem holds for the precision matrix. 
\begin{theorem}[\cite{caner2022sharpe}]
Under the main assumptions it holds that
\begin{align}
\underset{ 1 \leq j \leq p }{ \mathsf{max} } \ \norm{ \widehat{\boldsymbol{\Gamma}}_j - \boldsymbol{\Gamma}_j }_1 \equiv \mathcal{O}_p \big( \bar{m} \ell_n \big) = o_p(1).    
\end{align}
and secondly, 
\begin{align}
\norm{ \widehat{\boldsymbol{\mu}} - \boldsymbol{\mu} }_{\infty} = \mathcal{O}_p \left( \mathsf{max} \left\{K \sqrt{ \frac{ ln(n) }{n} }, \sqrt{ \frac{ ln(p) }{n} }  \right\} \right) = o_p(1).      
\end{align}
\end{theorem}

\begin{remark}
Using a direct estimation procedure (i.e., without a pseudo-inverse approximation) for the precision matrix provides a faster convergence to the true population values of the portfolio performance metrics. However, regardless of the direct or indirect estimation of the precision matrix, a larger number of $p$ (nodes in the network) affects the error by a logarithmic factor. On the other hand, the estimation error also increases with the non-sparsity of the precision matrix, especially as the dimensionality of the matrix increases. Thus, in the case of a non-sparse precision matrix, we can only get consistency when $p << n$, that is, the number of nodes in the network is much smaller than the time series observations on which the statistical model is fitted on. A key assumption in the conventional framework is that, $\big( \boldsymbol{Y}_t, \boldsymbol{X}_t  \big)_{t=1}^n$ are both stationary and ergodic vectors.    
\end{remark}

\begin{lemma}
The following probability bound holds
\begin{align}
\mathbb{P} \left(  \underset{ 1 \leq j \leq p }{ \mathsf{max} } \ \underset{ 1 \leq \ell \leq p }{ \mathsf{max} } \left| \frac{1}{n} \sum_{t=1}^n u_{\ell, t} u_{j,t} - \mathbb{E} \big[ u_{\ell,t} u_{j,t}   \big] \right| > C \sqrt{ \mathsf{ln} (p) / n } \right) = \mathcal{O} \left( \frac{1}{p^2}  \right)  
\end{align}
when the covariance structure is based on a prespecified factor model. 
\end{lemma}

%%-------------------------------------------------------------------------%%
\newpage 

\begin{example}
\label{example2}
Consider a given factor structure for stock returns such that 
\begin{align}
\boldsymbol{Y} = \boldsymbol{B} \boldsymbol{X}  + \boldsymbol{U}     
\end{align}
Then, the covariance matrix of $\boldsymbol{Y}$ is expressed as below
\begin{align}
\boldsymbol{\Sigma}_y =  \boldsymbol{B} \mathsf{Cov} \left( \boldsymbol{f}_t \right) \boldsymbol{B}^{\prime} + \boldsymbol{\Sigma}_n.     
\end{align}
while the precision matrix, defined by $\boldsymbol{\Gamma}_y := \boldsymbol{\Sigma}_y^{-1}$, such that  
\begin{align}
\boldsymbol{\Gamma}_y := \boldsymbol{\boldsymbol{\Omega}}_n - \boldsymbol{\boldsymbol{\Omega}}_n^{-1} \boldsymbol{B} \left[ \mathsf{Cov} \left( \boldsymbol{f}_t \right)^{-1} + \boldsymbol{B}^{\prime}   \right]     
\end{align}
where $\boldsymbol{\Omega}_n = \boldsymbol{\boldsymbol{\Sigma}}_n^{-1} \equiv \left\{ \mathbb{E} \left[ \boldsymbol{u}_t \boldsymbol{u}_t^{\prime} \right] \right\}^{-1}$.
\end{example}

Notice that Example \ref{example2} illustrates the decomposition of the precision matrix when the vector of stock returns have a factor structure representation. On the other hand, in this article we assume that the there exists an underline network structure that describes the stochastic behaviour and comovements of these nodes in the graph. Some key points are summarized below:

\begin{itemize}

\item Our proposed risk matrix is equivalent to the covariance matrix $\boldsymbol{\Sigma}_y$, after fitting the nodewise quantile predictive regression models in the system, while $\boldsymbol{\Omega} := \left\{ \mathbb{E} \left[ \boldsymbol{u}_t \boldsymbol{u}_t^{\prime} \right] \right\}^{-1}$, i.e., the precision matrix corresponds to the inverse of the covariance matrix of errors. However, the advantage of our approach is that we focus directly on the inversion of this tail risk matrix by assuming that $\boldsymbol{\Gamma}_{j_1 j_2 }^{-1} = \frac{ \omega_{j_1 j_2} }{ \omega_{j_1 j_1}  \omega_{j_2 j_2 }  }$. A novelty in the construction of this risk matrix is that the estimation of these risk measures is conditioned in the presence of covariates (i.e., stationary or nonstationary regressors) which is more informative than based solely on conditional distribution functionals. 

\item  In terms of the relevant testing methodology this involves finding statistical significant blocks based on the tail interdependence can be interpreted as revealing blocks which are significantly important using a form of pooling these extreme events together based on the risk measures of $\mathsf{VaR}$ and $\mathsf{CoVaR}$.  Although, we do not study this here, the role of the degree of persistence of these covariates included in the models, is an interesting future research question. 
    
\end{itemize}

\subsection{Related Literature.}

Extreme events (such as financial crises) are commonly employed for causal discovery since the relationships between variables may manifest themselves especially in the largest observations. In particular, a modelling methodology based on extreme value theory is presented in the study of \cite{gnecco2021causal}\footnote{In the particular framework the authors use $X_k := \sum_{ k \in pa (j, G) } \beta_{jk} X_k + \epsilon_j, \ j \in V$ under the assumption that the noise sequence $\epsilon_j, j \in V$, are heavy-tailed with extreme value index $\gamma > 0$ such that $P ( \epsilon_j > x  ) \sim \ \ell(x)x^{-2/ \gamma}, \ \ x \to \infty$.}. A discussion about tail dependence in graphs is provided in the paper of \cite{engelke2020graphical} (see also discussion in  \cite{dombry2016asymptotic}). On the other hand, our regression-based risk matrix takes into utilized tail estimates rather than relying on fitting a parametric model to heavy-tailed data, which implies that the corresponding matrix estimator is robust to the presence of potential outliers in financial and macroeconomic data.

%%-------------------------------------------------------------------------%%
\newpage 

The estimation of covariance matrices is widely used in various applications such as optimal portfolio choice problems and other statistical problems such as multivariate analysis, principal components analysis and graphical modeling among others. The commonly used assumption for the estimation of large covariance matrices is the assumption of multivariate normality. For instance, \cite{brown1975techniques} and \cite{browne1984asymptotically} proposed statistical methodologies for the analysis of covariance matrices such as asymptotically distribution-free methods. In particular, the majority of estimators of parameters in structural models for covariance matrices are members of the class of generalized least squares (GLS) estimators. Specifically, this type of estimators include the maximum Wishart likelihood estimators (MWL) in the sense that a generalized least squares discrepancy function can be constructed which will in general attain its minimum at the point defined by the MWL estimator (see, \cite{browne1984asymptotically}). A recent study who consider the MLE estimation of covariance type matrices when both rows and columns can be correlated is presented by \cite{drton2021existence}.     

The statistical literature on methodologies for the robust estimation of large covariance matrices has been thrived the last decades. In particular, assumptions such as the Wishart distribution for the unknown covariance matrix allows to implement asymptotically distribution-free methods for the robust estimation. However, features such as heavy-tailed time series, temporal dependence or even the presence of structural instability in time-series affects the robustness of the estimation procedure. For instance, \cite{dendramis2021estimation} accommodate the time variation, dependence and heavy-tailedness of distributions with implementation of a time-varying covariance matrix estimation procedure that includes thresholding. Therefore, the novelty of our estimation procedure is the use of quantile predictive regression models for capturing these features, which implies a regression-based tail risk matrix. Moreover, the particular proposition allows to take into consideration the time series properties via the estimation methodology rather than implementing ex-ante corrections to the estimation procedure which implies implementing bias corrections for the possible effect of the particular features to the robustness and degree of accuracy of the large covariance matrix. 

A different stream of literature investigates the use of centrality measures (see, a discussion in \cite{aamari2021graph}) for the development of graph-based modelling methodologies. Our approach is different than portfolio sorting procedures such as the studies of \cite{mcgee2020optimal} and \cite{ledoit2018efficient} who 
propose methodologies for sorting portfolios based on variable characteristics and nonlinear time series models respectively, however our motivation is driven by the construction  of an endogenous to the system sorting mechanism. Furthermore, the aspect of bridging centrality and extremity is examined in the study of \cite{einmahl2015bridging}. In terms of rank centrality measures a related approach is presented in the study of \cite{negahban2017rank}, while a formal asymptotic theory framework based on the random matrix theory is given by \cite{Li2021central}. Furthermore, \cite{negahban2017rank} propose a rank centrality measure which is an iterative rank aggregation algorithm for discovering scores for objects (or items) from pairwise comparisons in a statistical sense. 

A third line of literature which is related to our study is the use of the concept of \textit{conditional independence}. Specifically, a novel conditional independence measure is proposed by \cite{azadkia2021simple} which corresponds to the case of conditional mean preserving measures. Lastly, although we mainly focus with the aspects of statistical estimation, relevant inference procedures include testing for misspecification of the underline covariance structure (see, \cite{guo2021specification} and \cite{chang2022testing}).  Further limit theorems on covariance structures are given in the recent study of \cite{liu2021robust} while \cite{chang2021central} establish central limit theorems for high-dimensional dependent data.

%%-------------------------------------------------------------------------%%
\newpage

In summary, our proposed modelling approach focuses on two key aspects: \textit{(i)} the construction of the risk matrix based on tail dependence measures and \textit{(ii)} the implementation of the feature selection algorithm. Although our feature selection algorithm is applied to the tail risk matrix, we conjecture that this estimation procedure can be also used in alternative structures of covariance matrices. 

%A different stream of literature investigates related limit results to statistical inference under network dependence, which is beyond the scope of the current study\footnote{A relevant discussion to the challenging problem of conducting valid inference in modelling settings under network dependence is presented by \cite{lee2021network}.}. 

%\paragraph{Structure.} The paper is organized as follows. Section \ref{Section2} first introduces the regression-based tail risk matrix and then establishes our main result on the rate of convergence of the regression-based matrix estimator. Section \ref{Section3} presents our novel centrality measures based on the proposed regression-based matrix.

%\paragraph{Notation.} Throughout this paper, we use bold capital letters, for example $\boldsymbol{A}$ to represent matrices and bold letters in lowercase, such that $\boldsymbol{a}$, to denote vectors.  

%%-------------------------------------------------------------------------%%
%\newpage 

\section{The regression-based tail risk matrix.}
\label{Section2}

Specifically, the construction of the proposed risk matrix captures causality\footnote{A discussion related to statistical causality can be found in \cite{schenone2018causality}.} for extreme values based on nodewise quantile predictive regression models which capture the quantile behaviour of the underline distributions. Our novel risk matrix is constructed using nodewise regression models in the sense that the elements of the risk matrix are estimated using pairwise quantile regression models.

\subsection{Modelling Environment.}

 Given the singular value decomposition $\boldsymbol{\Gamma} = \boldsymbol{U} \boldsymbol{D} \boldsymbol{V}^{\prime}$. Moreover, denote with $\big\{ \boldsymbol{\Gamma}_t \big\}_{t=0}^T$ be the sequence obtained by the iteration of Algorithm 1. 

\bigskip

\begin{remark}(Initialization and the Stopping Rule) We suggest to initialize the algorithm with $\boldsymbol{\Gamma}_0$ in Algorithm 1, but because the optimization problem is convex, this can be replaced by any matrix. Furthermore, the algorithm converges faster when it is initialized with a matrix that is close to the minimizer. Furthermore, we suggest to stop the algorithm at iteration $T$ satisfying the condition $\big| F( \boldsymbol{\Gamma}_{T+1} ) - F( \boldsymbol{\Gamma}_{T} ) \big| \leq \epsilon$, for some small $\epsilon > 0$. 

\end{remark}

\medskip

Furthermore, in this section we derive some useful oracle inequalities such as provide suitable error bounds for the difference between the sequence of matrices $\boldsymbol{\Gamma}_n$ generated by Algorithm 1 and the true matrix $\boldsymbol{\Gamma}$. These results are determined by the strong convexity of the optimization function $\rho_{\uptau}$. 

\begin{assumption}
We impose the following conditions. 
 
\begin{itemize}

\item[\textbf{(A1)}] There exists $C > 0$ such that for $u_{ij} = Y_{ij} - \boldsymbol{X}_{i}^{\prime} \boldsymbol{\Gamma}_{.j}$ such that it holds that
\begin{align}
\mathbb{P} \big( | u_{ij} | > s \big) \leq \mathsf{exp} \left(  1 - \frac{s^2}{C^2} \right), \forall \ s \geq 0
\end{align}
with sub-gaussian norm given by 
\begin{align}
\norm{ u_{ij} }_{ \psi_2 } \overset{ \mathsf{def} }{ = } \underset{ p \geq 1 }{ \mathsf{sup} } \ p^{-1/2} \big( \mathbb{E} | u_{ij} |^p \big)^{1 / p} \ \ \text{and} \ \ K_u \overset{ \mathsf{def} }{ = } \underset{ 1 \leq j \leq m }{ \mathsf{max} } \ \norm{ u_{ij} }_{ \psi_2 }. 
\end{align}

\item[\textbf{(A2)}] Conditional on $\boldsymbol{X}_i$, $Y_{ij}$ are independent over $j$. 

\end{itemize} 
\end{assumption}

Notice that Assumption A3 is required in order to obtain the bounds on the tail probabilities of the estimation error. Furthermore, we consider the non-asymptotic bound for $\norm{ \boldsymbol{\Gamma}_t - \boldsymbol{\Gamma} }_F$ in the general situation that the number of factors can be increasing with $n$.

%%-------------------------------------------------------------------------%%
\newpage

\subsection{Estimation Methodology.}

Our risk matrix is employed to model tail causality for the nodewise regression models. In particular, these nodewise regressions correspond to quantile predictive regression models as defined by \cite{adrian2016covar} which are used to obtain estimates for the risk measures of VaR and CoVaR. Consider that we have the vector $\underline{Y} = ( Y_{1,t},..., Y_{p,t} )$ such that $i \in \left\{ 1,..., p \right\}$ and $j \in \left\{ 1,..., p \right\}$.  Consider the single index model for estimating the $\mathsf{CoVaR}$ and $\mathsf{VaR}$ risk measures
\begin{align}
Y_{i,t} &= \mu_i + \gamma_i  X_{t-1} + u_{i,t}
\\
Y_{j,t} &= \mu_{j|i} + \gamma_{j|i} X_{t-1} + \beta_{j|i} Y_{i,t} + u_{j|i, t}
\end{align}

% Therefore, by operating within the framework of moderate and large deviations principles we ensure that the estimation procedure is versatile enough to accommodate features such as nonstationarity in the time series properties of regressors. The second aspect of interest are the risk measures employed for the construction of the particular risk matrix. 

% Thus, our method can be examined from the estimation perspective which implies that we can be estimated with the use of these models which are wonderful models and have very nice properties.

% The inteconnectedness in the network is assumed to be driven by tail events based on a cross-section of firms that are linked with the pairwise quantile predictive regression models (see, \cite{adrian2016covar} and \cite{lee2016predictive}). 

% Moreover, we could define the particular matrix for a general structure of regression models not necessarily the quantile predictive regression model. This way we can focus on the proposed feature sorting algorithm. 

% In particular, we operate under the assumption of stationary time series data such that the pair $\left\{ y_t, x_{t-1} \right\}_{t=1}^n$ is stationary and ergodic where $1 \leq t \leq n$. Furthermore, notice that our framework we operate under the assumption that the data observations can be independent. In this sense we exclude cases such as spurious associations due to network dependence. More specifically, we explore the presence of network dependence due to the endogenous formation of a network modelled with the pairwise quantile predictive regression models. 

The given specification is motivated from the literature of financial connectedness and tail risk and in particular is based on the seminal works of \cite{adrian2016covar} and \cite{hardle2016tenet} who examine statistical estimation methodologies and their properties. We firstly consider the following covariance-type structure  
\begin{align}
  \mathbf{\widetilde{\Sigma}}_{ij} =  \begin{bmatrix}
    \text{VaR}_{1|1} &  \text{CoVaR}_{1|2}  & \dots \text{CoVaR}_{1|N} \\
    \vdots & \ddots & \vdots  \\
    \text{CoVaR}_{N|1} &  \dots & \text{VaR}_{N|N} \\
\end{bmatrix}
\end{align}
where  $\mathbf{\widetilde{\Sigma}}_{ij} \in \mathbb{R}^{ N \times N }$ represents the financial connectedness matrix and provides a robust representation of idiosyncratic and systemic risk within the financial network.  Our proposed risk matrix belongs to the family of dynamic high dimensional covariance matrices and captures network tail risk dependence, since is based on tail risk measures. We estimate the time-varying VaRs and CoVaRs conditional on a vector of lagged state variables $\boldsymbol{X}_{t-1}$. Both VaR and CoVaR risk measures represent quantiles of the distribution of returns under different conditioning sets; hence, to obtain one-period ahead predicted values for these risk quantities we employ the quantile regression specifications proposed by the seminal study of \cite{Koenker1978regression} (see, also \cite{KoenkerXiao02}). The construction of one-period ahead forecasts for the dynamic portfolio allocation, based on a graph representation, consists one of the main contributions of the empirical application of the paper. 

%The particular estimation procedure is employed for each of the pairwise predictive regressions, to obtain the elements of the VaR-$\Delta$CoVaR matrix. Next, we examine  the model fit and specification of the proposed risk matrix $\widetilde{ \mathbf{\Gamma} }$. 

Consider the conditional quantile function estimated via $Q_{y_i} \left( \tau | x_i \right) = F_{y_i}^{-1} \left( \tau | x_i \right)$. Then, the optimization function to obtain the model estimates is expressed as below
\begin{align}
\label{quantile.estimation}
Q_{\tau} \left( y_i | x_i \right) = \underset{ q(x)}{\text{arg min}} \ \mathbb{E} \big[ \rho_{\tau} \left( y_i - q(x_i) \right) \big]
\end{align}
where $\tau \in (0,1)$ is a specific quantile level, and $\rho_{\tau}( u ) = u \left( \tau - \mathbf{1}_{ \left\{ u < 0 \right\} } \right)$ is the check function (see, \cite{newey1987asymmetric}). Suppose $1 \leq t \leq T$, then based on the  estimation method given by \eqref{quantile.estimation} the VaR and CoVaR are generated by the following econometric specifications 
\begin{align}
R_{(i),t} &= \nu_{(i)}  +  \boldsymbol{\xi}_{(i)}^{\prime} \boldsymbol{X}_{t-1} + u_{(i),t} \label{VaR} \\
R_{(j),t} &= \nu_{(j)} + \boldsymbol{\xi}_{(j)}^{\prime} \boldsymbol{X}_{t-1} + \beta_{(j)} R_{(i),t}  + v_{(j),t} \label{CoVaR}
\end{align} 
where $\{ R_{(i),t} \}$ for $i \in \left\{ 1,...,N \right\}$ is the vector of portfolio returns at time $t$, and $\boldsymbol{X}_{t-1}$ is a vector of exogenous regressors containing macroeconomic characteristics common across assets and denote with $\Theta = \{ \nu_{(i)}, \boldsymbol{\xi}_{(i)}, \nu_{(j)}, \boldsymbol{ \xi }_{(j)}, \beta_{(j)} \}$ to be a compact parameter space.

%%-------------------------------------------------------------------------%%
\newpage 

Let $Q_{\tau} \left( \cdot\ | \ \mathcal{F}_{t-1} \right)$ denote the quantile operator for $\tau \in (0,1)$ conditional on an information set $\mathcal{F}_{t-1}$. Then, the error terms satisfy $Q_{\tau} \left( u_{(i),t} \ | \ \boldsymbol{X}_{t-1} \right)=0$ for the quantile regression model \eqref{VaR}, and $Q_{\tau} \left( v_{(j),t} \ | \ R_{(j),t}, \boldsymbol{X}_{t-1} \right) = 0$ for the quantile regression model given by expression \eqref{CoVaR}. Moreover, denote with $\boldsymbol{\theta}_{(1)} = \left( \nu_{(i)}, \boldsymbol{\xi}_{(i)}^{\prime} \right)$ and $\widetilde{\boldsymbol{X} }_{t-1}^{(1)} = \left( 1, \boldsymbol{X}_{t-1}^{\prime} \right)^{\prime}$ for model \eqref{VaR}, and with $\boldsymbol{\theta}_{(2)} = \left( \nu_{(j)}, \boldsymbol{\xi}_{(j)}^{\prime}, \beta_{(j)} \right)$ and  $\widetilde{\boldsymbol{X}}_{t-1}^{(2)} = \left( 1, \boldsymbol{X}_{t-1}^{\prime}, R_{(i),t} \right)^{\prime}$ for model \eqref{CoVaR}. 

Then, the model estimate 
\begin{align}
\widehat{ \boldsymbol{\theta} }_{(1)} \left( \tau  \right) = \underset{ \boldsymbol{\theta}_{(1)} \in \mathbb{R}^{p+2} }{  \text{arg min} } \sum_{t=1}^{T} \rho_{\tau} \left( R_{(i),t} - \boldsymbol{\theta}_{(1)}^{\prime} \widetilde{\boldsymbol{X} }_{t-1}^{(1)}  \right)
\end{align}
Similarly, the quantile estimator of \eqref{CoVaR} is obtained via the optimization function below 
\begin{align}
\widehat{ \boldsymbol{\theta} }_{(2)} \left( \tau  \right) = \underset{ \boldsymbol{\theta}_{(2) } \in \mathbb{R}^{p+1} }{  \text{arg min} } \sum_{t=1}^{T} \rho_{\tau} \left( R_{(j),t} - \boldsymbol{\theta}_{(2)}^{\prime} \widetilde{\boldsymbol{X} }_{t-1}^{(2)}  \right)
\end{align}
The state variables, that is, the macroeconomic and financial variables, included in the econometric model, capture time variation in the conditional moments of asset returns. Therefore, the $\mathsf{VaR}$ and $\mathsf{CoVaR}$ of each financial institution are estimated as vectors from a time-varying distribution. 

This approach allows to explain the optimal portfolio allocation in terms of changes in the systemic risk and financial connectedness in the network. To implement this, we use a rolling window, within which the quantile regressions given by  \eqref{VaR} and \eqref{CoVaR} are fitted and the $\boldsymbol{\Gamma}$ risk matrix is constructed based on one-period ahead forecasts. Specifically, the one-period ahead $\mathsf{VaR}_{t+1}$ for a coverage probability $\tau \in (0,1)$ is obtained by regressing  $R_{(i),t}$ on $\boldsymbol{X}_{t-1}$ as given by the first step regression (\ref{VaR}) which gives the parameter estimates $\widehat{\nu}_{(i)}$ and $\widehat{ \boldsymbol{\xi} }_{(i)}$. The forecasted one-period ahead VaR is computed as below 
\begin{align} \label{equ:VaR}
\widehat{ \mathsf{VaR} }_{(i),t+1} ( \tau ) = \widehat{\nu}_{(i)} + \widehat{ \boldsymbol{\xi}}^{\prime}_{(i)} \boldsymbol{X}_{t}.  
\end{align}

Similarly, the one-period ahead forecast of $\mathsf{CoVaR}_{t+1}$ is obtained by regressing $R_{(j),t}$ on $\boldsymbol{X}_{t-1}$ and $R_{(i),t}$ as shown in the second step regression (\ref{CoVaR}) using the quantile estimation as defined by \eqref{quantile.estimation}. To do this, we collect the parameter estimates $\widehat{\nu}_{(j)}, \widehat{\boldsymbol{\xi} }_{(j)}, \widehat{\beta}_{(j)}$ from the second regression, and construct the forecasted one-period ahead $\mathsf{CoVaR}_{t+1}$ measure using the following expression: 
\begin{align}
\label{equ:CoVaR}
\widehat{ \mathsf{CoVaR} }_{(j),t+1} ( \tau ) = \widehat{\nu}_{(j)} + \widehat{\mathbf{ \xi  }}_{(j)}  \mathbf{M}_{t} +  \widehat{\beta}_{(j)} \widehat{\mathsf{VaR}}_{i,t+1} ( \tau ),
\end{align}
where $\widehat{\mathsf{VaR}}_{i,t+1}( \tau )$ replaces the actual $\mathsf{VaR}_{i,t+1} ( \tau )$ measure.
Then, the $\mathsf{\Delta \text{CoVaR}}$ measure is computed as the difference of the two risk measures
\begin{equation}
\Delta \mathsf{CoVaR} = \widehat{\mathsf{CoVaR}}_{(j),t+1} ( \tau ) - \widehat{\mathsf{VaR}}_{(j),t+1 } ( \tau ),
\end{equation}
An important aspect of our proposed estimation methodology is the assumption of a graph representation with nodes being the stock returns and edges being pairwise tail forecasts, which can be thought as the level of simultaneous Granger causality due to the existence of tail and graph dependence. Therefore, these tail forecasts are estimated using the aforementioned predictive regression models based on the bipartite graph structure, which implies an estimation implementation for each pair $(i,j)$ of nodes for the graph $\mathcal{G} = (\mathcal{V}, \mathcal{E})$.

%%-------------------------------------------------------------------------%%
\newpage 

\subsection{Main Results}

\subsubsection{Asymptotic theory for the VaR-$\Delta$CoVaR risk matrix}

In this section, we derive results for establishing the asymptotic theory for the estimator of the proposed risk matrix that correspond to the population risk measures. Define with $\mathtt{Upper} ( \mathbf{\Gamma} ) \in \mathbb{R}^{ N \times N}$ the upper triangular part of the matrix, excluding the diagonal entries and with $\mathtt{Lower} ( \mathbf{\Gamma} ) \in \mathbb{R}^{ N \times N}$ the lower triangular part of the matrix, excluding the diagonal entries. Then, the risk matrix $\mathbf{\Gamma} \in \mathbb{R}^{ N \times N}$ is given by the following expression: 
\begin{small}
\begin{align*} 
\mathbf{\Gamma}  
:=  
\begin{bmatrix}
\mathsf{VaR}^{+}_{1}  &  (\mathsf{VaR}^{+}_{1} \mathsf{VaR}^{+}_{2})^{1/2} \Delta \mathsf{CoVaR}_{(1,2)} & \dots (\mathsf{VaR}^{+}_{1} \mathsf{VaR}^{+}_{N})^{1/2} \Delta \mathsf{CoVaR}_{(1,N)} \\
\vdots       &  \vdots        & \vdots             \\
\vdots       &  \ddots        & \vdots              \\
(\mathsf{VaR}^{+}_{N} \mathsf{VaR}^{+}_{1})^{1/2} \Delta \mathsf{CoVaR}_{(1,N)} &  \dots         & \mathsf{VaR}^{+}_{N}       \\
\end{bmatrix}
\end{align*}
\end{small}
where the off-diagonal elements consist of the term:
\begin{align}
\Delta \mathsf{CoVaR}_{(i,j)} = \mathsf{CoVaR}_{(i,j)} - \mathsf{VaR}^{+}_{i}
\end{align}
%Therefore the proposed risk matrix, $\mathbf{\Gamma} \in \mathbb{R}^{ N \times N}$ can be decomposed as below: 
%\begin{align}
%\mathbf{\Gamma}_{ N \times N } \equiv \left\{ \mathtt{Lower} ( \mathbf{\Gamma} ) + \mathsf{diag} \left( \mathsf{VaR}^{+}_{1},..., \mathsf{VaR}^{+}_{N} \right) + \mathtt{Upper} ( \mathbf{\Gamma} ) \right\}
%\end{align}
%where the risk matrix $\mathbf{\Gamma}_{ N \times N } $ is assumed to be positive-definite.

Since the main diagonal of the risk matrix  includes only the VaR tail measure which correspond to the underline node specific distribution, then we suppose that is generated by a common cumulative distribution function $F(.)$. An important condition for our risk matrix to be utilized within the optimal portfolio problem is to ensure that it converges in probability to a positive definite matrix for large sample size, that is, in our case, as the number of nodes in the network grows to infinity, i.e., $N \to \infty$.    

\medskip

Then, we can observe that the proposed covariance structure has a similar decomposition to the covariance matrix as below:   
\begin{small}
\begin{align*}
\mathbf{\Gamma} \equiv 
\begin{pmatrix}
\sqrt{ \mathsf{VaR}^{+}_{1} } & & & 
\\
& \sqrt{ \mathsf{VaR}^{+}_{2} }
\\
& & \ddots
\\
& & & \sqrt{ \mathsf{VaR}^{+}_{N} }
\end{pmatrix}
\begin{pmatrix}
1 & \Delta \mathsf{CoVaR}_{1,2}  & & 
\\
\Delta \mathsf{CoVaR}_{2,1} & 1
\\
& & \ddots
\\
& & & 1
\end{pmatrix}
\begin{pmatrix}
\sqrt{ \mathsf{VaR}^{+}_{1} } & & & 
\\
& \sqrt{ \mathsf{VaR}^{+}_{2} }
\\
& & \ddots
\\
& & & \sqrt{ \mathsf{VaR}^{+}_{N} }
\end{pmatrix}
\end{align*}
\end{small}
Moreover, we are interested in estimating the precision matrix of the $\mathsf{VaR}-\Delta \mathsf{CoVaR}$ risk matrix, $\boldsymbol{\Omega} := \boldsymbol{\Gamma}^{-1} \equiv \displaystyle \frac{ \omega_{j_1, j_2} }{ \omega_{j_1,j_1} \omega_{j_2,j_2} }$ where $\omega_{j_1, j_2}$ correspond to the elements of the precision matrix of the $\mathsf{VaR}-\Delta \mathsf{CoVaR}$ matrix for some $(j_1, j_2) \in \left\{ 1,..., N \right\}$. 

The asymptotic behaviour of the risk matrix can be determined by considering the limiting distribution of the expression $
\sqrt{N} \left( \widehat{ \mathbf{\Gamma} } - \mathbf{\Gamma}_0 \right)$. Denote with 
\begin{align}
\mathcal{A}_1 &:= \sqrt{N} \left(  \mathtt{Lower} ( \widehat{ \mathbf{\Gamma} } ) - \mathtt{Lower} ( \mathbf{\Gamma} )  \right)
\\
\mathcal{A}_2 &:= \sqrt{N} \left(  \mathtt{Upper} ( \widehat{ \mathbf{\Gamma} } ) - \mathtt{Upper} ( \mathbf{\Gamma} )  \right)
\\
\mathcal{A}_3 &:= \sqrt{N} \left\{  \text{diag} \left( \widehat{\mathsf{VaR}}^{+}_{1},..., \widehat{ \mathsf{VaR} }^{+}_{N}  \right) - \text{diag} \left( \mathsf{VaR}^{+}_{1},..., \mathsf{VaR} ^{+}_{N} \right)  \right\}
\end{align}
The term $\mathcal{A}_3$ converges to a Gaussian random variable, with a non-stochastic covariance matrix under the assumption of stationary and ergodic regressors in the quantile predicitve regression model. 

%%-------------------------------------------------------------------------%%
\newpage 

%Consider a vector $\boldsymbol{B}$ such that
%\begin{align}
%\boldsymbol{B} = \left( \frac{ \lambda_1 }{ 1 - \rho_1 },..., \frac{ \lambda_m }{ 1 - \rho_m } \right)^{\top}    
%\end{align}
%and
%\begin{align}
%\boldsymbol{V}_{j, \ell} 
%= 
%\begin{cases}
%\gamma_j^2 & \ \text{if} \ j = \ell,
%\\
%\gamma_j \gamma_{\ell} & \ \text{if} \ j = \ell 
%\end{cases}

%Then, it holds that 
%\begin{align}
%\bigg( \sqrt{k_1} \left( \widehat{\gamma}_1 ( k_1 ) - \gamma_1  \right),..., \sqrt{k_m} \left( \widehat{\gamma}_m ( k_m ) - \gamma_m  \right)  \bigg)^{\top} \overset{d}{\to} \mathcal{N} \big( \boldsymbol{B}, \boldsymbol{V} \big).   
%\end{align}
%In particular, if $\boldsymbol{\gamma} = \gamma \boldsymbol{1}$, then it holds that 
%\begin{align}
%\sqrt{k} \left( \widehat{\gamma}_n ( \omega ) - \gamma \right) \overset{d}{\to} \mathcal{N} \big(  \boldsymbol{\omega}^{\top} \boldsymbol{B}_{ \boldsymbol{c} }, \boldsymbol{\omega}^{\top} \boldsymbol{V}_{ \boldsymbol{c} } \boldsymbol{\omega} \big),    
%\end{align}
%with 
%\begin{align}
%\boldsymbol{B}_{ \boldsymbol{c} } = \sqrt{ \sum_{i=1}^m c_i^{-1} } \left(  \sqrt{ c_1 } \frac{ \lambda_1 }{ 1 - \rho_1 },..., \sqrt{ c_m }  \frac{ \lambda_m }{ 1 - \rho_m } \right)^{\top}     
%\end{align}

\medskip

\begin{example}
As an illustration of the required derivations to obtain asymptotic theory results, consider the case in which the graph matrix depends on the unknown error terms of $\boldsymbol{\Gamma} := [u_{ij}]$. We follow similar derivations as in the framework of \cite{guo2021specification}. Define with $\mathcal{G}_1 := \big\{ i : (i,j) \in \mathcal{G}, \ \text{for some} \ j \leq p \big\}$. For any $(k \times k)$ matrix $\boldsymbol{A}$, denote by $\lambda_i (\boldsymbol{A})$ the $i-$th largest eigenvalue of $\boldsymbol{A}$ for $i = 1,..., k$. Then, some useful results are given below. 
\end{example}

\medskip

\begin{lemma}
Under the assumptions of the Theorem, there exists $\zeta_1 > 0$ and $\zeta_2 > 0$, such that
\begin{align}
\mathbb{P} \left(  \underset{ (i,j) \in \mathcal{G} }{ \mathsf{max} } \ T^{1/2} \left|  \frac{1}{T} \sum_{t=1}^T \hat{u}_{it} \hat{u}_{jt} - \frac{1}{T} \sum_{t=1}^T u_{it} u_{jt} \right| \geq \zeta_1 \right) < \zeta_2   
\end{align}
where $\zeta_1 = o \big( ( \mathsf{log}(p) )^{-1/2} \big)$ and $\zeta_1 = o(1)$.
\end{lemma}

\begin{proof}
Since it holds that 
\begin{align}
\sum_{t=1}^T \big( \hat{u}_{it} \hat{u}_{jt} - u_{it} u_{jt} \big)  
= 
\sum_{t=1}^T \bigg\{ \big( \hat{u}_{it} - u_{it} \big) \big( \hat{u}_{jt} - u_{jt} \big)   + u_{it} \big( \hat{u}_{jt} - u_{jt} \big) + \big( \hat{u}_{it} - u_{it} \big) u_{jt} \bigg\}
\end{align}
Therefore, it suffices to show that 
\begin{align}
\underset{ (i,j) \in \mathcal{G} }{ \mathsf{max} } \left| \sum_{t=1}^T  \big( \hat{u}_{it} - u_{it} \big) \big( \hat{u}_{jt} - u_{jt} \big) \right| = o_p \left( T^{1/2} \zeta_1 \right)   
\end{align}
In particular, it holds that 
\begin{align*}
\underset{ (i,j) \in \mathcal{G} }{ \mathsf{max} } \left| \sum_{t=1}^T  \big( \hat{u}_{it} - u_{it} \big) \big( \hat{u}_{jt} - u_{jt} \big) \right| 
&\leq
\underset{ (i,j) \in \mathcal{G} }{ \mathsf{max} } \sum_{t=1}^T \left| \big( \hat{u}_{it} - u_{it} \big) \big( \hat{u}_{jt} - u_{jt} \big) \right|
\\
&\leq
\underset{ (i,j) \in \mathcal{G} }{ \mathsf{max} } \left\{  \sum_{t=1}^T \big( \hat{u}_{it} - u_{it} \big)^2 . \sum_{t=1}^T \big( \hat{u}_{jt} - u_{jt} \big)^2 \right\}^{1/2}
\\
& = 
\underset{ i \in \mathcal{G}_1 }{ \mathsf{max} } \left\{ \sum_{t=1}^T \big( \hat{u}_{it} - u_{it} \big)^2 \right\}^{1/2}  \underset{ j \in \mathcal{G}_2 }{ \mathsf{max} } \left\{ \sum_{t=1}^T \big( \hat{u}_{jt} - u_{jt} \big)^2 \right\}^{1/2}
\\
& =
\underset{ i \in \mathcal{G}_1 }{ \mathsf{max} } \left\{ \sum_{t=1}^T \left[ \big( b_i^{\prime} - \hat{b}_i^{\prime} \big) z_t \right]^2 \right\}^{1/2}  \underset{ j \in \mathcal{G}_2 }{ \mathsf{max} } \left\{ \sum_{t=1}^T \left[ \big( b_j^{\prime} - \hat{b}_j^{\prime} \big) z_t \right]^2 \right\}^{1/2}
\end{align*}
Consider the following expression 
\begin{align*}
M_{ij} 
:= \frac{1}{T} \sum_{t=1}^T \left\{ u_{it} u_{jt} - \phi_{ij} ( \theta_0) 
- \frac{ \partial \phi_{ij} ( \theta)  }{ \partial \theta^{\top}    } \bigg|_{\theta = \theta_0 }  \left[  \sum_{(k, \ell) \in \mathcal{H} } \frac{ \partial \phi_{k \ell} (\theta) }{  \partial \theta } \frac{ \partial \phi_{k \ell} (\theta) }{  \partial \theta^{\top} }  \right]^{-1}  \times \sum_{(k, \ell) \in \mathcal{H} } \bigg[ u_{kt} u_{\ell t} - \phi_{k \ell } (\theta_0)     \bigg] \frac{ \partial \phi_{k \ell} (\theta) }{ \partial \theta }  \bigg|_{\theta = \theta_0 }  \right\}    
\end{align*}
where $\theta^{*}$ lies between $\hat{\theta}$ and $\theta_0$.

%%-------------------------------------------------------------------------%%
\newpage

Specifically, we will show that 
\begin{align}
\mathbb{P} \left(  \underset{ (i,j) \in \mathcal{H} }{ \mathsf{max} }         \ T^{1/2} \left|  \phi_{ij} \left( \hat{\theta} \right) - \phi_{ij} ( \theta_0 ) - M_{ij}^* \right| \geq \zeta_1 \right) < \zeta_2,
\end{align}
where 
\begin{align*}
M_{ij}^{*} = \frac{1}{T} \sum_{t=1}^T \left( \frac{ \partial \phi_{ij} (\theta) }{  \partial \theta^{\top} } \bigg|_{\theta = \theta_0 } \left[  \sum_{(k, \ell) \in \mathcal{H} } \frac{ \partial \phi_{k \ell} (\theta) }{  \partial \theta } \frac{ \partial \phi_{k \ell} (\theta) }{  \partial \theta^{\top} }  \right]^{-1}  \sum_{ ( k,\ell ) \in \mathcal{H} } \left\{ \bigg[ u_{kt} u_{\ell t} - \phi_{k \ell } (\theta_0)  \bigg] \frac{ \partial \phi_{k \ell} (\theta) }{ \partial \theta }  \bigg|_{\theta = \theta_0 } \right\} \right)    
\end{align*}
We know that
\begin{align}
\underset{ (i,j) \in \mathcal{H} }{ \mathsf{max} } \norm{ \frac{1}{2} \left( \hat{\theta} - \theta_0 \right)^{\top} \frac{ \partial^2 \phi_{ij} (\theta) }{ \partial \theta  \partial \theta^{\top} } \bigg|_{ \theta = \theta^* }  \left( \hat{\theta} - \theta_0 \right)  }_2 \leq \norm{ \hat{\theta} - \theta_0  }_2^2  = \mathcal{O}_p \left(  \frac{ \mathsf{log} (p) }{T} \right).    
\end{align}
because of the definition of $\hat{\theta}$, which implies that 
\begin{align}
\sum_{ (i,j) \in \mathcal{H} } \big( \hat{\sigma}_{ij} - \phi_{ij} (\hat{\theta} ) \big)  \frac{ \partial \phi_{ij} (\theta) }{ \partial \theta} \bigg|_{ \theta = \hat{\theta} } = 0_q.
\end{align}
From Taylor expansion, we have that 
\begin{align*}
0_q 
&= 
\sum_{ (i,j) \in \mathcal{H} } \left[ \left( \hat{\sigma}_{ij} - \phi_{ij} (\hat{\theta} ) \right)  \frac{ \partial \phi_{ij} (\theta) }{ \partial \theta} \bigg|_{ \theta = \theta_0 } \right]  
\\
&+
\left( \sum_{ (i,j) \in \mathcal{H} } \left[ \left( \hat{\sigma}_{ij} - \phi_{ij} ( \theta^* ) \right)  \frac{ \partial^2 \phi_{ij} (\theta) }{ \partial \theta \partial \theta^{\top}  } \bigg|_{  \theta = \theta^*  } \right] - \sum_{ (i,j) \in \mathcal{H} } \left[ \frac{ \partial \phi_{ij} (\theta) }{  \partial \theta} \bigg|_{ \theta = \theta^* } \frac{ \partial \phi_{ij} (\theta) }{ \partial \theta^{\top} } \bigg|_{ \theta = \theta^* } \right] \right) \left( \hat{\theta} - \theta_0    \right)
\end{align*}
where $\theta^*$ lies between $\hat{\theta}$ and $\theta_0$. Therefore, it holds that 
\begin{align}
\hat{\theta} - \theta_0 = N^{-1} \sum_{ (i,j) \in \mathcal{H} } \left[ \left( \hat{\sigma}_{ij} - \phi_{ij} ( \theta_0 ) \right)  \frac{ \partial \phi_{ij} (\theta) }{ \partial \theta} \bigg|_{ \theta = \theta_0 } \right]     
\end{align}

\end{proof}

\begin{remark} 
Notice that it is important to distinguish the source of high-dimensionality. In particular, for $\beta_j, j \in \left\{ 1,..., p \right\}$, we can consider asymptotic expressions in the limit where both $n, p \to \infty$ and focus on the high-dimensional regime with $c^{*} = \mathsf{lim} \ p / n \in (0, \infty]$. Moreover, our estimation method differs from various current methodologies in the literature. For example, in the study of \cite{gonzalo2021spurious} the authors are concerned with the interactions of persistence and dimensionality in the context of the eigenvalue estimation problem of large covariance matrices arising in cointegration and principal components analysis. In other words, we are not concerned about the presence of spurious relationships since our covariance matrix is constructed based tail estimates and not the error terms of cointegrating systems. On the other hand, there is a clear interplay between the dimensionality of the statistical problem of interest and granger  causality.    
\end{remark}

%To add related eigenvalue conditions since we consider a high dimensional problem. In particular, such assumptions are common in the high-dimensional statistics literature. The main complexity is that, rather than for the original design matrix, the restricted eigenvalue condition is imposed on the transformed design matrices. 

%\item A second aspect to be explained is the implementation of a predictive matrix and the main difference between a predictive model and a predictive regression model. 

%%-------------------------------------------------------------------------%%
\newpage

\subsection{Optimal Portfolio Optimization.} 
\label{SectionI}

We begin with the traditional optimal portfolio allocation problem. We extend this framework by proposing a novel risk matrix that captures tail events. We present regularity conditions which ensure the existence of a closed-form solution for the associated quadratic risk function and optimization functions. Our proposed risk matrix induces a well defined quadratic form which permits to be employed for optimal portfolio choice problems\footnote{A rational investor who holds a set of securities aims to minimize the portfolio risk based on the portfolio optimization framework introduced by \cite{markowitz1956optimization}. The specific linear optimization mechanism has seen growing attention in various fields of financial economics such as portfolio management, risk management, capital investment and asset pricing. Furthermore, the seminal paper of \cite{sharpe1964capital} introduced the fundamental principles of investment under conditions of uncertainty in financial markets. Furthermore, various studies examined the effects of constraints on the portfolio performance measures such as the expected return and the portfolio risk (see \cite{fishburn1977mean}, \cite{jagannathan2003risk}, and references therein).}.

\subsubsection{Traditional Minimum-Variance Portfolio Optimization.}
\label{SectionI.A}

The optimal portfolio theory proposed by \cite{Markowitz52} considers the minimization of the portfolio variance without any further restriction on the corresponding expected return. More formally, let $\mathbf{R}_t = ( R_{1,t} ,..., R_{N,t} )^{\prime}$ be the $N-$dimensional vector of assets returns at time $t$ and denote with $\text{\boldmath$\mu$} := \mathbb{E} \left( \mathbf{R}_t \right) = \left( \mathbb{E} \left( R_{1,t} \right) ,..., \mathbb{E} \left( R_{N,t} \right) \right)^{\prime}$ the vector of expected returns and  $\mathbf{\Sigma} := \text{Cov}\left( \mathbf{R}_t \right) = \mathbb{E} \left( \mathbf{R}_t \mathbf{R}_t^{\prime} \right) - \text{\boldmath$\mu$} \text{\boldmath$\mu$}^{\prime}$ the associated variance-covariance matrix. Let $r_t^{p} = \mathbf{w}_i^{\prime} \mathbf{R}_t$ denote the return on a portfolio of $N$ assets, with $\mathbf{w}_i= \left( w_{1},...,w_{N} \right)^{\prime}$ the vector of portfolio weights representing the financial position of the investor as a weight allocation of assets to the portfolio. The variance of this portfolio is $V \left( r_t^{p} \right) = \mathbf{w}^{\prime} \mathbf{\Sigma} \mathbf{w}$. Then, the Markowitz's minimum variance optimization problem is expressed as below
\begin{align}
\label{the problem}
\underset{\mathbf{w} \in \mathbb{R}^N }{\text{arg min}} \ \big\{  \mathbf{w}^{\prime} \mathbf{\Sigma} \mathbf{w} \big\} \ \text{subject to }  \  \mathbf{w}^{\prime} \mathbf{1}= 1.
\end{align}
The first order conditions to the minimization problem yield the optimal portfolio allocation given by the following vector of weights
\begin{align}
\label{minv}
\mathbf{w}^{*} = \frac{\mathbf{\Sigma}^{-1} \mathbf{1}}{\mathbf{1}^{\prime} \mathbf{\Sigma}^{-1} \mathbf{1} }.
\end{align}

\subsubsection{Portfolio Allocation and Asset Centrality.}

\begin{align} 
\label{matrix.new.specification}
Q \left( \widehat{ \boldsymbol{\Sigma} },\mathbf{w} \right)= \mathbf{w}^{\prime} \widehat{ \boldsymbol{\Sigma} } \mathbf{w}, 
\end{align}

The symmetry of $\widehat{ \mathbf{\Sigma} }$ yields the spectral decomposition $\widehat{ \mathbf{\Sigma} } = \mathbf{U} \mathbf{D}_{\sigma} \mathbf{U}^{\prime}$, with $\mathbf{U}$ is an $N \times N$ orthonormal matrix such that $\mathbf{U}^{\prime} = \mathbf{U}^{-1}$ which includes as columns the linearly independent eigenvectors of $\widehat{ \mathbf{\Sigma} }$, and $\mathbf{D}_{\sigma}$ an $N \times N$ diagonal matrix with the corresponding eigenvalues which is defined as $\mathbf{D}_{ \sigma } = diag \left\{ \lambda^{ \sigma }_1,..., \lambda^{ \sigma }_N \right\}$. Moreover, we can rearrange the diagonal eigenvalue matrix as $\mathbf{D}_{\sigma} =\mathbf{U}^{\prime} \widehat{ \mathbf{\Sigma} } \mathbf{U}$.  

%%-------------------------------------------------------------------------%%
\newpage 

This allow us to express the eigenvalues of the risk matrix  $\widehat{ \mathbf{\Sigma} }$ in terms of the elements of the variance-covariance matrix as below
\begin{align} 
\label{lambda1}
\lambda^{ \sigma }_{k} = \overset{N}{\underset{i=1}{\sum}}  u_{ik}^2 \sigma^{2}_{k} + \overset{N}{\underset{i=1}{\sum}}  \overset{N}{\underset{\underset{j \neq i}{j=1}}{\sum}} u_{ik} u_{jk} \sigma_{ij}, \ \ \text{for} \ k \in \left\{ 1,...,N \right\}.
\end{align}
where $\left\{ u_{ij} \right\}_{i,j = 1,...,N}$ denotes the eigenvectors of the risk matrix $\widehat{ \mathbf{\Sigma} } \in \mathbb{R}^{N \times N}$. 
\begin{assumption} 
\label{assumption1}
The eigenvectors $\{ u_{ij} \}_{i,j=1,...,N}$ of the risk matrix $\mathbf{\Gamma}$ satisfy the condition 
\begin{equation}
\overset{N}{\underset{i=1}{\sum}} u_{ik}^2 > - \overset{N}{\underset{i=1}{\sum}}  \overset{N}{\underset{\underset{j \neq i}{j=1}}{\sum}} u_{ik} u_{jk} \frac{ \widetilde{ \gamma}_{i|j} }{  \gamma_{k}^{0} },
\end{equation}
for all $k=1,\ldots,N$.
\end{assumption}
The condition in Assumption \ref{assumption1} guarantees that $\widetilde{ \mathbf{\Gamma} }$ is positive definite, and as a result implies that the corresponding quadratic form $\mathbf{w}^{\prime} \widetilde{ \mathbf{\Gamma} } \mathbf{w}$ has certain desirable properties as summarized by the following propositions.  

\begin{proposition}
\label{propositon2}
Let $Q( \widetilde{ \mathbf{\Gamma} } ,\mathbf{w}) = \mathbf{w}^{\prime} \widetilde{ \mathbf{\Gamma} } \mathbf{w}$ denote the portfolio risk, where the risk matrix $\widetilde{ \mathbf{\Gamma} }$ satisfying Assumption \ref{assumption1}. Then the following conditions hold
\begin{itemize}
\item[(i)] $Q( \widetilde{ \mathbf{\Gamma} },\mathbf{w})= 0$ if and only if $\mathbf{w}=0$.
\item[(ii)] $Q( \widetilde{ \mathbf{\Gamma} },\mathbf{w}) > 0$ for $\mathbf{w} \neq 0$.
\item[(iii)] The associated quadratic form defined by the Lagrangian of the objective function \eqref{minVaR} given by $Q( \widetilde{ \mathbf{\Gamma} },\mathbf{w}) + \zeta (\mathbf{w}^{\prime} \mathbf{1} - 1)$, for some $\zeta > 0$,  has a global minimum on $\mathbf{w}$.
\end{itemize}
\end{proposition}

\medskip

\begin{remark}
Notice that the optimal choice of weights goes beyond the portfolio allocation optimization problem. For example, the choice of $\boldsymbol{w}$ can be based on the vector of naive weights or a vector estimated through an optimization procedure. For example, the covariance matrix used in portfolio allocation problems is positive-definite which allow to estimate the vector of optimal weights. We provide some examples below:
\end{remark}

\begin{itemize}

\item[$\boldsymbol{1.}$] $\textbf{(Variance-Optimal Weights)}$. There is a unique solution to the minimization problem of $\boldsymbol{w}^{\top} \boldsymbol{V} \boldsymbol{w}$ subject to the constraint that $\boldsymbol{w}^{\top} \boldsymbol{1} = 1$, which is
\begin{align}
\boldsymbol{w}_{gmvp}  = \frac{ \boldsymbol{V}_c^{-1} \boldsymbol{1} }{ \boldsymbol{1}^{\top} \boldsymbol{V}_c^{-1} \boldsymbol{1} }, \ \sqrt{k} \big( \widehat{\gamma}_n \left( \boldsymbol{w}_{gmvp} \right)  - \gamma \big) \overset{d}{\to} \mathcal{N} \left( \frac{ \boldsymbol{1}^{\top} \boldsymbol{V}_c^{-1} \boldsymbol{B}_c }{ \boldsymbol{1}^{\top} \boldsymbol{V}_c^{-1} \boldsymbol{1} },  \frac{1}{ \boldsymbol{1}^{\top} \boldsymbol{V}_c^{-1} \boldsymbol{1} }         \right).  
\end{align}

\item[$\boldsymbol{2.}$] $\textbf{(AMSE-Optimal Weights)}$. There is a unique solution to the minimization problem of $\mathsf{AMSE}(\boldsymbol{w}) =  \displaystyle \frac{1}{k} \left[ \left( \boldsymbol{w}^{\top} \boldsymbol{B}_c \right)^2 + \boldsymbol{w}^{\top} \boldsymbol{V}_c \boldsymbol{w} \right]$ subject to the constraint $\boldsymbol{w}^{\top} \boldsymbol{1} = 1$. 
\end{itemize}
Lastly, if $\boldsymbol{w}_n \boldsymbol{n} = 1$ with $\widehat{\boldsymbol{w}}_n \overset{p}{\to} \boldsymbol{w}$, then the composite estimator $\widehat{\gamma}_n ( \widehat{\boldsymbol{w}}_n )$ is $\sqrt{k}-$asymptotically equivalent to $\widehat{\gamma}_n ( \boldsymbol{w}_n )$ in the sense that $\sqrt{k} \big( \widehat{\gamma}_n ( \widehat{\boldsymbol{w}}_n ) - \widehat{\gamma}_n ( \widehat{\boldsymbol{w}}_n ) \big) = o_p(1)$.

%%-------------------------------------------------------------------------%%
\newpage

\newpage

\section{A Novel Centrality Measure.}
\label{Section3}

\subsection{Graph Topology and Node Centrality.}
\label{SectionII.A}

We denote with $ \mathcal{G} = \{ \mathcal{V} , \mathcal{E} \}$ a graph structure, representing a financial network, which consists of a set of nodes, $\mathcal{V}$, and a set of edges, $\mathcal{E}$, connecting the pairs of nodes. A full characterization of the information in the network is provided by the $N \times N$ adjacency matrix $\mathbf{\Omega}^0$ whose element $\{ \Omega_{ij}^0 \}_{i,j=1,...,N}$ determine whether there is a link connecting node $i$ and $j$ for the graph $\mathcal{G}$. The link between two elements can be characterized by a binary variable which determines the existence of a connection or by a real value that corresponds to the pairwise weight between nodes. 

We propose the following adjacency matrix $\mathbf{\Omega}^0 = \widetilde{ \mathbf{\Gamma} } - diag \left( \widetilde{ \mathbf{\Gamma} } \right)$, where $diag(\cdot)$ is an $N \times N$ diagonal matrix, implying that $diag \left( \widetilde{ \mathbf{\Gamma} } \right) :=  \text{diag} \{ \gamma^{0}_1, \ldots, \gamma^{0}_N \}$. Therefore, within our framework we consider a weighted network, where the connections between stocks are determined by the $\widetilde{ \gamma}_{i|j}$ measures as defined by the off-diagonal elements of the $\widetilde{ \mathbf{\Gamma} }$ matrix. The adjacency matrix $\mathbf{\Omega}^0$ is symmetric by definition, since the main diagonal has a vector of zeros and the off-diagonal terms are equivalent to the symmetric risk matrix $\widetilde{ \mathbf{\Gamma} }$, that is, $\Omega_{ij}^0= \widetilde{ \gamma}_{i|j}$ for $i \neq j$. 

The symmetry of the adjacency matrix entails the spectral decomposition $\mathbf{\Omega}^0 = \mathbf{Z}_{ \Omega^0 } \mathbf{D}_{ \Omega^0 } \mathbf{Z}_{ \Omega^0 }^{\prime}$, with $\mathbf{Z}_{ \Omega^0 }$ an $N \times N$ orthonormal matrix such that $\mathbf{Z}_{ \Omega^0 }^{\prime} = \mathbf{Z}_{ \Omega^0 }^{-1}$ that contains the linearly independent eigenvectors of $\mathbf{\Omega}^0$, and $\mathbf{D}_{\Omega^0}$ is an $N \times N$ diagonal matrix with the corresponding eigenvalues. Moreover, we express the eigenvalues of the adjacency matrix $\mathbf{\Omega}^0$ using the spectral vector decomposition as below   
\begin{align} 
\label{lambda1b}
\lambda_{k}^{ \Omega^0 } =\sum_{i=1}^{N} \overset{N}{\underset{\underset{j \neq i}{j=1}}{\sum}} z_{ik} z_{jk} \widetilde{ \gamma}_{i|j},
\end{align}
with $z_{ij}$ the eigenvectors of the matrix $\mathbf{\Omega}^0 \in \mathbb{R}^{N \times N}$. We decompose the loss function $\mathbf{w}' \widetilde{ \mathbf{\Gamma} } \mathbf{w}$  as the sum of a quadratic loss function given by the adjacency matrix and a quadratic loss function given by a diagonal risk matrix with elements given by $\gamma^{0}_i$. More formally,
the quadratic form becomes
\begin{equation} 
\label{decom}
Q( \widetilde{ \mathbf{\Gamma} } ,\mathbf{w}) = \mathbf{w}' \widetilde{ \mathbf{\Gamma} }
 \mathbf{w} \equiv \mathbf{w}' \mathbf{\Omega}^0 \mathbf{w} + \sum_{i=1}^{N} w_{i}^2 \gamma^{0}_i.
\end{equation}
The study of the contribution of the centrality of assets to the loss function $Q( \widetilde{ \mathbf{\Gamma}} , \mathbf{w})$ is examined by \cite{katsouris2021optimal} and \cite{olmo2021optimal}. To do so, we further examine the components of expression $\mathbf{w}' \mathbf{ \Omega }^0 \mathbf{w}$ and provide a definition of asset centrality with respect to the eigenvalue decomposition of the adjacency matrix. The notion of centrality quantifies the influence of certain nodes in a given network. There are several measurements in the literature each corresponding to a specific definition of centrality. We focus on the eigenvector centrality, see also \cite{Bonacich72}. This measure of centrality is closely related to the concept of Katz centrality, see \cite{katz1953new}. The eigenvector centrality of asset $i$ is defined as the proportional sum of its neighbors' centrality, see \cite{Bonacich72}.

%%-------------------------------------------------------------------------%%
\newpage

\subsection{Main Theory.}

\begin{itemize}

\item Let $\succsim$ be a non-trivial binary relation over $X$. We say that $\succsim$ has a weak representation if there exists a non-constant real valued function $v : X \to \mathbb{R}$ such that for all $x, y \in X, v(x) > v(y)$ implies $x \succ y$.

\item We say that $v$ is a partial representation for $\succsim$ if for all $x,y \in X, x \succ y$ implies $v(x) > v(y)$. 

\item The function $v$ is called strong representation of $\succsim$ if for all $x, y \in X, x \succ y$ if and only if $v(x) > v(y)$, that is, $v$ is a strong representation if it is both a weak and partial representation. 

\end{itemize}

A centrality measure although has no formal definition in the literature in this paper we propose some regularity properties for these measures to correspond to centrality ranking. Motivated from the above findings we propose a novel centrality measure. In particular, in this paper we define the context of recursive monotonicity. In mathematical terms recursive monotonicity means that if vertex $i$ has more neighbors than vertex $j$, and $i$'s neighbours are more central than those of $j$, then $i$ must be more central than $j$. Therefore, we focus on two different network topologies, a topology of low connectivity and a topology of strong connectivity (see, \cite{sadler2022ordinal}). 
\color{black}

A measure of centrality is a function $C$ that takes a node $i \in V$ and the graph $G = ( V, E)$ it belongs to, and returns a non-negative real number $C(i ; G) \geq 0$ which translates to a quantity which specifies how central $i$ is in $G$. Currently, there is no agreement on a formal definition of a centrality measure the following properties should hold

\begin{itemize}

\item \textit{Invariance}.  

\item \textit{Monotonicity}. Adding an edge between $i$ and another node $j$, then the centrality of $i$ does not increase.  

\end{itemize} 

\begin{assumption}[Centrality measures] Under the assumption of a graph $\mathcal{G} = \left( E, V \right)$ it holds:

\begin{itemize}

\item[\textit{\textbf{(a)}}] \textit{Linear Homogeneity}. Removing one node from the graph, then the properties of the centrality measure remain valid since it provides a ranking statistic in relation to the graph topology, that is, the ranked association of the node $i \in \left\{ 1,..., N \right\}$ with respect to the remaining nodes.  

\item[\textit{\textbf{(b)}}] \textit{Invariance}. The centrality measure remains invariant under a linear transformation. A linear (affine) transformation of the adjacency matrix does not affect the ranking statistics in relation to the graph topology. Furthermore, the centrality measures are invariant under graph automorphisms, that is, nodes re-labeling. 

\item[\textit{\textbf{(c)}}] \textit{Recursive Monotonicity}.  A relevant definition of recursive monotonicity is given by \cite{bommier2017monotone}.  

\item[\textit{\textbf{(d)}}] \textit{Sub-Additivity}. 

\end{itemize}

\end{assumption}

Furthermore and more precisely, given a preorder on the vertex set of a graph, the neighborhood of vertex $i$ dominates that of another vertex $j$ if there exists an injective function from the neighbourhood of $j$ to that of $i$ such that every neighbor of $j$ is mapped to a higher ranked neighbour of $i$. 

%%-------------------------------------------------------------------------%%
\newpage 

The preorder satisfies recursive monotonicity, and is therefore an ordinal centrality, if it ranks $i$ higher than $j$ whenever $i$'s neighbourhood dominates that of $j$. Consider for example commonly used centrality measures in the literature such as degree centrality, Katz-Bonacich centrality and eigenvector centrality, all of these centrality measures satisfy assumption 1 (recursive monotonicity).     
Therefore, following our preliminary empirical findings of the optimal portfolio performance measures with respect to the low centrality topology as well as the high centrality topology, we obtain some useful insights regarding the how the degree of connectedness in the network and the strength of the centrality association can affect the portfolio allocation. Building on this idea, in this section we focus on the proposed algorithm which provides a statistical methodology for constructing a centrality ranking as well as a characterization of the strength of the centrality association between nodes. 

More specifically, we assume that a node exhibits "strong centrality" when all the rank centrality statistics assign a higher preference to vertex $i$ above $j$ (in terms of graph topology), which implies that in a sense $i$ is considered to be more central than the node $j$ for all $i,j \in \left\{ 1,..., N \right\}$ such that $i \neq j$. On the other hand, we assume that a node exhibits "weak centrality" when vertex $i$ is weakly more central than vertex $j$ if there exists some strict ordinal centrality that ranks $i$ higher than $j$. 

According to \cite{sadler2022ordinal} weak centrality itself is a strict ordinal centrality, so it is the maximal such order, and it extends strong centrality. Moreover, in this paper we employ these terminologies for the development of an algorithm for feature ordering by centrality exclusion.
We define some useful mathematical definitions such as the binary relation $\succsim$ on $V$ is a set of ordered pairs of vertices, and we write $ i \succsim j$ if $(i,j) \in \succsim$.  

\begin{definition}[recursive monotonicity]
Given a graph $G = ( V, E)$ and a binary relation $\succsim$ on $V$, the subset $S \subseteq V$ $\mathbf{\succsim}-$dominates $S^{\prime} \subseteq V$, and we write $S \succsim S^{\prime}$, if there exists an injective function $f: S^{\prime} \to S$ such that $f(i) \succsim i$ for all $i \in S^{\prime}$. Furthermore, the binary relation $\succsim$ satisfies \textit{\textbf{recursive monotonicity}} if $i \succsim j$ whenever $G_i \succsim G_j$. 
\end{definition}

\begin{remark}
An ordinal centrality $\succsim$ on $G = (V, E)$ is a preorder on $V$ that satisfies recursive monotonicity. Recursive monotonicity captures the essence of a board class of centrality measures. Thus a vertex $i$ is more central than vertex $j$ if $i$ has more connections to more central nodes than $j$. 
\end{remark}

\begin{definition}[preorder]
The function $c: V \to \mathbb{R}$ represents the preorder $\succsim$ if we have $i \succsim j$ if and only if $c(i) \succsim c(j)$. 
\end{definition}

\begin{definition}[strongly more central]
Given a graph $G = ( V, E )$, node $i$ is \textit{\textbf{strongly more central}} than node $j$, and we write $i \succsim_s j$, if and only if $i \succsim j$ for every ordinal centrality on $G$.
\end{definition}

%\color{magenta}
%\begin{definition}
%A function $f$ from $S \subset \mathbb{R}^n$ into $\mathbb{R}^m$ is Lipschitz continuous at $x \in S$ if there is a constant $C$ such that 
%\begin{align}
%\norm{ f(y) - f(x) } \leq C \norm{ y - x }
%\end{align}
%for all $y \in S$ sufficiently near $x$. 
%\end{definition}

%\begin{remark}
%Notice that Lipschitz continuity at a point depends only on the behaviour of the function near that point. For f to be Lipschitz continuous at $x$, the inequality result must hold for all $y$ sufficiently near $x$, but it is not necessary that the above relation hold if $y$ is not near $x$.  
%\end{remark}

%%-------------------------------------------------------------------------%%
%\newpage 

\subsection{Centrality Rank Statistic.}

We follow the example presented in \cite{janssen2003bootstrap}. Given $k(n)$ let $\big( R_1,..., R_{k(n)} \big)$ be uniformly distributed ranks, that is, a random variable with uniform distribution on a set of permutations $\mathcal{P}_{k(n)}$ of $\left\{ 1,..., k(n) \right\}$. Furthermore, an additional index $n$ concerning the ranks is suppressed throughout. Consider a sequence of simple linear permutation statistics such that 
\begin{align}
S_n = \sum_{i=1}^{k(n)} \big( c_{n_i} - \bar{c}_n \big) d_n(R_i)
\end{align}
where $c_{n_i}$ are regression coefficients, $1 \leq i \leq k(n)$, with 
\begin{align}
\sum_{i=1}^{k(n)} \big( c_{n_i} - \bar{c}_n \big)^2 = 1, \ \ \ \ \bar{c}_n := \frac{1}{k_n} \sum_{i=1}^{k(n)} c_{n_i}
\end{align}
and random scores $d_n(i): \tilde{\Omega} \mathbb{R}$, $1 \leq i 
\leq k(n)$ with 
\begin{align}
\frac{1}{k(n) - 1} \sum_{i=1}^{k(n)} \big( d_n(i) - \bar{d}_n  \big)^2 = 1 \ \ \ \text{on} \ \ \ \mathcal{A}_n := \left\{ \sum_{i=1}^{k(n)} \big( d_n(i) - \bar{d}_n \big)^2  > 0 \right\}
\end{align}
The $c'$s are here considered to be fixed and the $d'$ are allowed to be random variables independent of the ranks $R_j$. Without restrictions we may assume ordered regression coefficients
\begin{align}
c_{n_1} \leq c_{n_2} \leq ... \leq c_{n k(n)}.
\end{align}  
Then, we have equality in distribution for $\big( d_n (R_i)   \big)_{ i \leq k(n) } \overset{ D }{ = } \big( d_{ R_{i : k(n) } } \big)_{ i \leq k(n) }$, which is a consequence of the ex-changeability of these variables. Therefore, in these cases ranks and order statistics are independent. Thus, without restriction we may assume that the score functions are also ordered
\begin{align}
d_n(1) \leq d_n(2) \leq ... \leq d_n \big( k(n) \big). 
\end{align}  
Therefore, since $\mathcal{E} \big( S_n | W \big) = 0$ holds always convergence subsequences of $S_n$ exist and under some regularity conditions we can classify all possible limit distributions of subsequences with respect to the distributional convergence. We define the normalized random variables as below
\begin{align}
Y_{n,i} := \frac{ \displaystyle X_{n,i} - \bar{X}_n }{  \displaystyle \left\{ \sum_{i=1}^{ k(n) } \big( X_{n,i} - \bar{X}_n   \big)^2 \right\}^{1/2} }. 
\end{align}

\medskip

%\begin{lemma}
%Let $\xi_n$ be a distributionally convergence sequence of random variables with values in a complete separable metric space. Then, there exists almost surely convergent versions of $\xi_n$ on a separable probability space.   
%\end{lemma}

%%-------------------------------------------------------------------------%%
%\newpage 

\subsection{Feature Sorting by Centrality Exclusion.}

The proposed algorithm of exclusion centrality is constructed using the proposed high-dimensional tail forecast risk matrix. However, the feature ordering by centrality exclusion (FOCE) procedure in practise can be applicable for covariance matrices for graphical models. On the other hand, when the FOCE procedure is implemented based on the VaR-$\Delta$CoVaR risk matrix, then the feature selection is based on the tail dependence of the response variable $\boldsymbol{Y}_t$ given the predictors $\boldsymbol{X}_{t-1}$. Furthermore, since the construction of the proposed risk matrix is based on the quantile predictive regression models, then we specifically focus on a type of granger causality in the tails of the underline distributions.  Thus, our proposed algorithm provides a methodology for ordering the vertices of the graph based on the effect of their centrality in relation to the other nodes in the graph when we exclude the most central node in the graph. It produces an ordering of the predictors according to their predictive power. This ordering is used for variable selection without putting any assumption on the distribution of the data or assuming any particular underlying model. 

\newpage

The simplicity of the estimation of the conditional dependence coefficient makes it an efficient method for variable ordering and variable selection that can be used for high dimensional settings. In this section, motivated by the study of centrality measures and the proposed risk matrix of the paper, we propose a novel feature selection algorithm for multivariate regression models using the centrality exclusion procedure based on our novel regression-based risk matrix. Notice that other methodologies currently proposed in the literature include the model-based methods.   

Let $Y$ be the response variable and let $\boldsymbol{X} = \left( X_j   \right)_{ 1 \leq j \leq p }$ be the set of predictors. The data consists of $n$ \textit{i.i.d} copies of $\left( Y, \boldsymbol{X}   \right)$. Similar to the framework proposed by \cite{azadkia2021simple}, first, choose $j_1$ to be the index $j$ that maximizes $T_n( Y, X_j )$. Then, having obtained $j_1,..., j_k$, choose $j_{k+1}$ to be the index $j \not\in \left\{ j_1,..., j_k \right\}$ that maximizes $T_n \left( Y, X_j | X_{j_1},...., X_{j_k} \right)$. Furthermore, continue like this until arriving at the first $k$ such that  $T_n \left( Y, X_{j_{k+1} } | X_{j_1},...., X_{j_k} \right) \leq 0$, and then declare the chosen subset to be $\hat{S} := \big\{  j_1,..., j_k  \big\}$. If there is no such $k$, define $\hat{S}$ to be the whole set of variables. Furthermore, this may also happen that $T_n \big( Y, X_{j_1} \big) \leq 0$. 

In that case, we declare $\hat{S}$ to be the empty set. In the next section, we prove the consistency of FSCE under a set of assumptions on the law of $\left( Y, \boldsymbol{X} \right)$. Furthermore, we focus on proposing on a well-defined stopping rule. Our stopping rule seems to work well in practise, and we are able to prove consistency of variable selection for this rule.

Suppose that $\boldsymbol{X}$ is a normal random vector with zero mean and arbitrary covariance structure, 
\begin{align}
Y = \beta \boldsymbol{X} + \epsilon,
\end{align}
Then, for any non-empty $S \subset \left\{ 1,..., p \right\}$  and any $j \in \left\{ 1,..., p \right\} \backslash S$. 
Notice that if $S$ is a sufficient set of predictors, then $\rho( S, j ) = 0$ for any $j \not\in S$. To the best of our knowledge the proposed algorithm for higher-order conditional risk measures is novel and is based on distribution-free inference. 

See for reference the paper: \cite{chao2018multivariate}. In the particular paper the authors introduce the following estimation procedure. Let $\big\{ \left( \boldsymbol{X}_i, Y_{i1},..., Y_{im}   \right) \big\}_{ 1 \leq i \leq n }$ where $Y_{ij}$ represents the value observed from the response $j$ at the time point $i$, and $\left\{ \boldsymbol{X}_i \right\}_{i=1}^n$ are the covariates. Furthermore, we assume that the samples are \textit{i.i.d} over $i$. For $\uptau \in (0,1)$, the conditional expectile $e_j \left( \uptau | \boldsymbol{X}_i  \right)$ of $Y_{ij}$ given $\boldsymbol{X}_i$ is defined as below
\begin{align}
e_j \left( \uptau | \boldsymbol{X}_i  \right) = \boldsymbol{X}_i^{\prime} \boldsymbol{\gamma}_j ( \uptau ), 
\end{align}
where
\begin{align}
\boldsymbol{\gamma}_j (\uptau) \overset{ \mathsf{def}  }{ = } \ \underset{ \boldsymbol{\gamma} \in \mathbb{R}^p }{ \mathsf{arg \ min} } \ \mathbb{E} \big[ \rho_{\uptau} \left( Y_{ij} - \boldsymbol{X}^{\prime} \boldsymbol{\gamma} \right)  \big],
\end{align}

The proposed feature selection by node exclusion algorithm is considered as a dimension reduction methodology for the graph. Below we present some related algorithms from the economic theory and econometrics literature.

%\color{red}
%(to add and explain the procedure more formally)

%Below we provide an example from the literature. 
%\color{black}

\begin{itemize}

\item \textbf{Heuristic Pricing Algorithm:} The particular algorithm iteratively creates an ordering of the choice situations, which can correspond to a rational choice type. First, we explain the link between orderings of the choice situations and rational choice types. Next, we explain how to build an ordering that provides a good solution. 

\item \textbf{Generation of Choice Types for Tightening:} To tighten the set based on a subset of the rational choice types, we generate the subset in a semi-random way. First, we generate (likely irrational) choice types by randomly choosing one patch in each choice situation. If this choice type is rational, we add it to the subset for tightening. If it is not, we identify the subsets of choice situations for which preference cycle exists. For each such subset, we randomly pick one choice situation. For that choice situation, we look for a patch which (i) removes at least one preference relation within the subset, (ii) is as close as possible to the currently selected patch in that choice situation, and (iii) removes (rathern than adds) revealed preference relations. In this way, we slightly change the choice type, while increasing the probability that it is a rational choice type. If after these changes the choice type is not yet rational, the procedure is repeated until a rational choice type is found. The algorithm below contains the pseudo-code to generate these rational choice types in a semi-random way. In the algorithm, we again define $\mathcal{T} = \left\{ t | 1 \leq t \leq T \right\}$ as the set of all choice situations. 
\end{itemize}

%%-------------------------------------------------------------------------%%
%\newpage 

\subsection{Algorithm.}

We propose a novel feature ordering algorithm for a multivariate regression with a fixed number of covariates, using a backward stepwise algorithm which is applied to our network driven tail risk matrix. In other words, given a high dimensional response vector $\boldsymbol{Y} = ( Y_j )_{1 \leq j \leq p}$ and a common set of predictors $\boldsymbol{X} = ( X_j )_{1 \leq j \leq k}$ then, the high dimensional response vector is reduced to a lower dimensional space based on the same number of predictors. In other words, given a high dimensional response vector and a set of covariates, we conjecture that the dimension of the response vector can be reduced without having to extend the set of predictors in order to better explain variation in the high-dimensional response vector. 

\begin{example}
Consider a class of linear additive models
\begin{align}
y_t = B z_t + u_t, \ \ \ t = 1,..., T    
\end{align}
with $y_t = \big( y_{1t},..., y_{pt} \big)^{\prime}$ and $u_t = \big( u_{1t},..., u_{pt} \big)^{\prime}$
\end{example}

\begin{itemize}

\item Both $y_t$, the high-dimensional response vector, and $z_t$ are observable. Usually, the literature considers the estimation of $\Sigma$ associated with the unobserved error term $u_t$, which is a measure of uncertainty. On the other hand we consider the predictive ability using the one-ahead period forecasts. Our methodology reduces the dimensionality of the response vector for a fixed number of covariates at the quantile level. %On the other hand, under the assumption that $K$ diverges with $T$, penalized estimators can be applied. 

\item Notice that in practise we examined two different procedures here. The first procedure utilizes a suitable stopping rule with respect to a portfolio optimization problem. The second procedure, considers an alternative way of ordering the importance (with respect to the degree of centrality) of the high-dimensional response vector. %However, is this procedure a measure of conditional dependence or a procedure for variable selection? 

\item A relevant notion from the statistical theory perspective is the aspect of $\textit{sufficient dimension reduction in classical statistics}$. In the classical statistics setting, if one can find a small subset of predictors that is sufficient, then we can assume that these predictors contain all the relevant predictive information about $Y$ among the given set of predictors, and the statistician can then fit a predictive model based on this small subset of predictors. 

\item On the other hand, if the reduction of the dimensionality of the high dimensional response vector is what the statistician is interested to, then our approach can work regardless of the dimensionality of the predictors. This implies that a question of constructing a statistical procedure for testing the inclusion of high dimensional controls is replaced on whether a multivariate regression model should include or not a high dimensional response vector. 
\end{itemize}
The main idea of the proposed algorithm is a way of selecting features or reducing the number of nodes in a graph, since the algorithm chooses the most central node and then removes it from the graph iteratively using a stopping rule. In other words, the proposed procedure can be considered to be a methodology that bridge the gap between a node exclusion statistical methodology and a procedure for feature selection in graphs without employing Lasso regularization techniques.

\section{Numerical Studies.}
\label{Section4}

\subsection{Simulation study.}

In this section, we apply our method on the simulated data to evaluate the estimation performance on the factors and loadings, as the number of factor varies. 

We set with $n = m = p = 100$. For $i = 1,..., n, j = 1,...,m$, and let $\boldsymbol{X}_i \sim \mathcal{N} \big(  \boldsymbol{0}_{p \times 1}, \boldsymbol{\Sigma}_{p \times p} \big)$ with $\boldsymbol{\Sigma}_{ij} = 0.5^{ | j - k  | }$ and $\boldsymbol{\epsilon}_i \overset{ i.i.d }{ \sim } \mathcal{N} \big( \boldsymbol{0}_{m \times 1},  \boldsymbol{I}_{ m \times m} \big)$. Then, the response variables are generated by 
\begin{align}
Y_{ij} 
= 
\boldsymbol{X}_i^{\prime} \boldsymbol{\Gamma}_{.j} + \epsilon_{ij}
=
\sum_{k=1}^{r} \psi_{jk} f_k \left( \boldsymbol{X}_i \right)  + \epsilon_{ij} 
=
\sum_{k=1}^{r} \boldsymbol{V}_{jk} \boldsymbol{D}_{kk} \boldsymbol{X}_i^{\prime} \boldsymbol{U}_{.k} + \epsilon_{ij},   
\end{align}
where $r = \mathsf{rank} \left( \boldsymbol{\Gamma} \right)$. 

\color{black}

Let $\Gamma$ be the an $N \times N$ matrix whose $(i,j)$ entry is $\upgamma_{ij}$. Furthermore, let $\displaystyle \Gamma = \sum_{j=1}^N s_j u_j v_j^{\prime}$. Furthermore, notice that since the elements of the risk matrix $\Gamma$ are constructed based from population moments of the pairwise conditional distribution functions. Considering the fact that the entries of the risk matrix $\Gamma$ are estimated by the nodewise quantile predictive regression models. Thus, our matrix is considered to be a matrix with pairwise tail forecasts. 

\bigskip

Let $\Gamma$ be the $N \times N$ matrix whose $(i,j)-$th element is $f( \beta_i, \beta_j )$. Then, our data matrix $X$ has the following form 
\begin{align}
x_{ij} = f( \beta_i, \beta_j ) + \epsilon_{ij}
\end{align} 
where $\epsilon_{ij}$ are independent errors with zero mean, satisfying the restriction that $| x_{ij} | < 1$ almost surely. For example, X may be the adjacency matrix of a random graph where the probability of an edge existing between vertices $i$ and $j$ is $f( \beta_i, \beta_j )$.

%%-------------------------------------------------------------------------%%
\newpage

\section{Conclusion.}
\label{Section5}

Our contributions to the literature are twofold: \textit{(i)} we propose a novel regression-based risk matrix for modelling tail dependence based on graph structures; and \textit{(ii)} we propose a statistical procedure for feature selection in graphical models. Furthermore, we focus on the interpretation of the risk matrix for optimal portfolio allocation problems as well as a mechanism to introduce our novel variable selection algorithm. Our study discusses important aspects related to the robust estimation of large tail forecast covariance-type matrices as well as the implementation of feature ordering procedures using the proposed graph-based matrix.

Obviously, the drawback of our proposition is the related challenges in evaluating the predictive accuracy of the tail risk matrix since the elements of the matrix consist of risk measures such as the Value-at-Risk and Conditional-Value-at-Risk which are tail estimates of the underline distributions. Secondly, the formal study of our tail risk matrix as a suitable representation for optimal portfolio choice problems is crucial since this is a novel aspect not previously proposed in the literature. Our motivation in this paper although is from the optimal portfolio allocation perspective in financial networks, we focus on the aspect of the induced centrality measure for the proposed novel risk matrix.

Some future research worth mentioning include the formal study of the asymptotic properties of the proposed risk matrix. Although we assume that we employ stationary time series observations our framework can be extended to accommodate features of non-stationary time series. We leave the particular aspect for future research. A second aspect of interest are the time series properties of regressors when estimating the risk measures of VaR and CoVaR. Although the risk management literature usually operates under the assumption of stationarity, an interesting aspect for future research is the effect of non-stationarity when estimating our novel risk matrix. We leave the particular aspect for future research. Although in this paper, we do not consider a formal statistical methodology for edge exclusion this is certainly an aspect which worth future research in subsequent studies. Furthermore, our proposed framework can be employed as a methodology for covariate screening in high dimensional environments. We leave the particular aspect for future research.

\bigskip

\paragraph{Acknowledgements} 

\begin{small} I wish to thank my main PhD supervisor Prof. Jose Olmo for guidance and continuous encouragement throughout the PhD programme as well as Prof. Peter W. Smith for helpful discussion and for stimulating the further development of the current study in the direction of node exclusion and feature selection. In addition, I wish to thank Zudi Lu, Jean-Yves Pitarakis, Tassos Magdalinos, Hector Calvo Pardo, Zacharias  Maniadis, Ruben Sanchez Garcia and Tullio Mancini as well as Departmental Seminar Series speakers: Luis E. Candelaria, Juan Carlos Escanciano, Marcelo C. Medeiros, Sebastian  Engelke, Markus Pelger and Loriano Mancini for helpful conversations. 

The author acknowledge the use of Iridis 5 HPC Facility and associated support services at the University of Southampton in the completion of this work. 

\paragraph{Funding} Financial support from the Vice-Chancellor's PhD scholarship of the University of Southampton is gratefully acknowledged. The author declares no conflicts of interests. 
\end{small}

%%-------------------------------------------------------------------------%%
\newpage 

\section{Appendix A. Background Literature}

\begin{small}

\subsection{Centrality Measures from network theory}

Firstly, for the network representation we construct an adjacency matrix induced via the econometric specification (i.e., CoVaR regressions) of the VaR-CoVaR matrix. The particular pairwise stock-return to stock-return regressions of financial institutions give an economic interpretation to our proposed adjacency matrix and as a result the links between the nodes of the network can be  interpreted in terms of their level of financial connectedness. This approach allows to both test the degree of financial connectedness in the network, via the adjacency matrix which reflects the existence (in a statistical sense) of financial connectedness between the assets under tail events, and construct the network topology using centrality measures. Secondly, we propose a novel centrality measure which is applied directly on the financial connectedness matrix and thus provide a robust way of measuring the centrality of assets via the network tail driven risk matrix. Centrality measures are often used as a way to provide a statistical representation of how connected a node is and to access spillover effects within the network. The introduction of centrality measures can answer relevant questions within the framework of financial networks such as, "What is the most vulnerable to economic shocks node?" or "What is the level of interconnectedness of core versus periphery nodes?". Furthermore, these centrality measures can provide network information related to (i) the properties of local topology\footnote{For example, the weighted characteristic path length measures the average shortest path, from one node to any other node in the network. Thus, it can be interpreted as the shortest number of connections a shock from a node needs to reach another connected node in the network.} via measures such as degree centrality and page rank and (ii) global information such as closeness centrality and betweenness centrality. Let $\lambda_i$ and $v_i$ denote the eigenvalue and corresponding eigenvector of the adjacency matrix $\mathbf{A}$ for a set of nodes $N$ and let $\mathbf{1}$ represent the unit vector, i.e., $\mathbf{1} = (1,...,1)$. 

\subsubsection{Network centrality measures}

\begin{itemize}
    \item \textit{Katz centrality}. Proposed by \cite{katz1953new}, for a symmetric adjacency matrix $\mathbf{A}$ with a vector of centrality scores for the nodes of the network given by
\begin{align}
    KC_i(\alpha ) = [ ( \mathbf{I} - \alpha \mathbf{A}^{\top})^{-1} \mathbf{1}  ]_i = \sum_{j=0}^{\infty} \alpha^j \sum_{i=1}^N \lambda_i^k v_i v_i^{\top} \mathbf{1} 
\end{align}
The benefits of Katz centrality is that it the centrality score of nodes can be decomposed into two components, i.e., the idiosyncratic centrality and the system-related centrality (the centrality passed to it in proportion to how important its neighbours are). Due to the construction of katz centrality of capturing the influence of nodes (i.e., the nodes which a node is connected to) is considered as a robust centrality measure in capturing financial contagion and risk transmission since it captures the influence of financial institutions due to connectedness induced by the underline network topology. 

\item \textit{Page rank centrality}. A modification of the katz centrality is the page rank centrality, which corrects for the contribution of neighbouring nodes on the impact each node has within the network. More specifically, with the eigenvector and katz centrality, there is no distinction between degree centrality and the level of connectedness of these neighbouring nodes. For example, low degree nodes may receive a high score because they are connected to very high degree nodes, even though they may have low degree centrality. Let $d_j$ be the degree centrality of node $j$. Then, page rank centrality scales the contribution of node $i$'s neighbours, $j$, to the centrality of node $i$ by the degree of $i$ given by the following expression
\begin{align}
 PR_i = \alpha \sum_{j=1}^N A_{ji} \frac{v_j}{d_j} + \beta = \mathbf{\beta}(\mathbf{I} - \alpha\mathbf{D}^{-1} \mathbf{A})^{-1} 
\end{align}

%%-------------------------------------------------------------------------%%
\newpage

\item \textit{Closeness centrality}. It access the centrality of a node at the local neighbourhood level. For example, the larger the closeness centrality of an institution the faster the influence in the other nodes of the network since it requires fewer steps for an impact to reach other nodes. The normalized closeness centrality of a node is computed as
\begin{align}
   CC_i= \frac{N-1}{\sum_{j=1}^N d_{ij} }
\end{align}

\item \textit{Betweeness centrality}. Considered for example, two financial institutions which have large betweenness centrality, this implies that the pair is important is the transmission of shocks. It is defined as the ratio of the total number of all shortest paths in the network that go via this node and the number of all other shortest paths that do not pass this node. 
\begin{align}
\label{equ:beast}
  CB_i= \sum_{s \neq t \neq j} \frac{\sigma_{st}(i)}{\sigma_{st}}
\end{align}

\item \textit{Leverage centrality}. Leverage centrality considers the degree of a node relative to its neighbours and is based on the principle that a node in a network is central if its immediate neighbours rely on that node for information\footnote{A node with negative leverage centrality is influenced by its neighbors, as the neighbors connect and interact with far more nodes. A node with positive leverage centrality, on the other hand, influences its neighbors since the neighbors tend to have far fewer connections (e.g., see \cite{vargas2017graph} }). The leverage centrality is computed as 
\begin{align}
 LC_i =  \frac{1}{d_i} \sum_{i \in \mathcal{N}_i} \frac{ d_i - d_j }{ d_i + d_j } 
\end{align}

\item \textit{Eigenvector centrality}. The eigenvector of node $i$ is equal to the leading eigenvector $\mathbf{v}_i$ and is computed using the characteristic equation of the adjacency matrix. Thus, the EC is defined
\begin{align}
 \mathbf{v}_i = \sum_{ j \in N(i) } \mathbf{v}_j = \sum_j A_{ij}  \mathbf{v}_j
\end{align}
Thus, we can see that the above definition of the eigenvector centrality implies that it depends on both the number of neighbours $|N(i)|$ and the quality of its connections $\mathbf{v}_j$, for $j \in N(i)$.  

\end{itemize}

We mainly focus on the above centrality measures\footnote{Other centrality measures include for example, the Bonacich network centrality measure proposed by \cite{bonacich1987power} which counts the number of all paths (not just shortest paths) that emanate from a given
node, weighted by a decay factor that decreases with the length of these paths.} for constructing the network as explained in the next section. In terms of the optimal asset allocation problem, such centrality measures permits to access the sensitivity of assets to the network topology and thus the effect to the portfolio perfrormance\footnote{We should note also that when choosing multiple metrics to describe the network topology these may be correlated and so an appropriate normalization may be needed to have a complete set of orthogonal centrality metrics (see, e.g., \cite{van2014graph}).}.

%%-------------------------------------------------------------------------%%
\newpage 

\subsection{Variance-Covariance Matrix as Adjacency Matrix}

\begin{example}
Define the weighted adjacency matrix $\mathbf{\Omega} \equiv 1$ if $i=j$ and $\mathbf{\Omega} \equiv \rho_{ij}$ if $i \neq j$, then the covariance matrix of returns, $\mathbf{\Sigma} = [\sigma_{ij}]$ i.e., $\sigma_{ii}=\sigma_i^2 \ \forall i=j$ can be decomposed as below
\begin{align}
\mathbf{\Sigma} = \mathbf{D} \mathbf{\Omega} \mathbf{D} \ \ \text{where} \ \ \mathbf{D}= \text{diag} \big\{ \sigma_{11},...,\sigma_{NN} \big\}. 
\end{align} 
Also, the adjacency matrix can be decomposed further using the following decomposition
\begin{align}
\mathbf{\Omega} = \mathbf{P} \mathbf{\Lambda} \mathbf{P}^{\top} \ \ \text{where} \ \ \mathbf{\Lambda}= \text{diag}[ \sqrt{\lambda_{ii}}] 
\end{align}
where $\mathbf{P}= [\mathbf{v}_1,...,\mathbf{v}_N ]$ i.e., the columns of the orthogonal matrix $\mathbf{P}$ consists of the set of eigenvectors of the adjacency matrix. The above matrix decompositions induced specifically a correlation-based adjacency matrix and so the optimal weights are functions of the particular measures, inducing a correlation driven network representation of the interactions. These two decompositions allow to easily compute the inverse of the adjacency matrix as a function of its eigenvalue and eigenvectors which in turn can relate to the eigenvalue centrality, a measure of the network centrality, given by $\mathbf{\Omega}^{-1} = \sum_{k=1}^N \frac{1}{\lambda_k} \mathbf{v}_k \mathbf{v}^{\top}_k$.
\end{example}

More specifically, using the particular network representation, it can be proved that the optimal portfolio weights are functions of the centrality of the assets in the portfolio. 
\begin{align}
\mathbf{w}^{*}_{(s)} = \phi_{(s)} .\ \mathbf{m}_{(s)} + \phi_{(s)} \bigg(\frac{1}{ \lambda_1} - 1 \bigg) \mathbf{m}_M . \mathbf{v}_1 + \Gamma_{(s)}
\end{align}
where $s$ denotes the type of strategy, $\mathbf{m}_{(s)}$ corresponds to $\mathbf{\epsilon}$ in the case of minimum variance strategy and $\hat{\mathbf{\mu}}^{e}$ in the case of mean-variance strategy. The term $\Gamma_{(s)}$ represents a derived function of the eigenvalues and eigenvectors of the adjacency matrix. In the case of the minimum-variance strategy this is given by
\begin{align}
 \Gamma_{(minv)} = \displaystyle \phi_{(minv)} \bigg\{ \sum_{k=2}^N \bigg(\frac{1}{ \lambda_1} - 1 \bigg) \mathbf{v}_k \mathbf{v}^{'}_k \bigg\} \mathbf{\epsilon}
\end{align}
where $\phi_{(minv)}= \displaystyle \frac{1}{\mathbf{1}^{\top} \mathbf{\Sigma}^{-1} \mathbf{1}}$ and $\displaystyle  \lambda_{1} = \text{max}_{x \neq 0} \frac{x^{\top} \Omega x}{x^{\top}x}$ the largest eigenvalue of the adjacency matrix an important feature of the network topology. The eigenvector centrality of node $i$, used in the above formulation of the optimal weights, is a measure describing the network topology and is computed by $v_i = \lambda^{-1} \sum_{j=1}^N \Omega_{ij} v_j$. This measure induces the network topology by giving emphasis on highly connected nodes since a large value of $v_i$ corresponds to a node which is connected either to many other nodes or to just few highly central nodes.     

\begin{remark}
Based on the above formulations \cite{peralta2016network} find a negative relationship between optimal portfolio weights and the centrality of assets. Therefore, strongly embedded stocks in a correlation based network affect the market stability and thus the inclusion of such assets in a portfolio undermines the benefits of diversification resulting in larger variances. Generally, a network investment strategy according to the authors makes more efficient use of fundamental information, resulting in a substantial reduction of wealth misallocation. In other words, the current methodology simplifies the portfolio selection process by targeting a group of stocks within a certain range of network centrality.    
\end{remark}

Secondly, \cite{anufriev2015connecting}, provide also an examination of the financial connectedness of the Australian financial institutions. The particular study is close the literature regarding the construction of a financial network using time series data (e.g., \cite{billio2012econometric}, \cite{diebold2012better}, \cite{diebold2014network}, \cite{barigozzi2017network}). However, their approach is also close to \cite{peralta2016network} since the authors establish the various links of the financial dependencies using the metrics and centrality measures of the network theory, however the commonality of this stream of literature is that the network is induced via the use of aggregate information and publicly available information related to the financial conditions of the nodes and the economy as a whole. 

%%-------------------------------------------------------------------------%%
\newpage 

Moreover, \cite{gupta2017credit}, construct a financial network via the use of a loan network which allows to examine the credit market spillovers. This approach uses data on the firm level by combining information of corporate loans to US firms and capturing the dynamic evolution of the network via the connections between the loans of the firms in the network. The proposed construction of a partial correlation network of \cite{anufriev2015connecting} allows to observe how a unit variance shock may spread through the network with the use of an adjacency matrix which represents the partial correlations of share prices of financial institutions. The partial correlations are related to the unconditional covariance matrix and to the precision matrix or the inverse covariance matrix, i.e., $\mathbf{K} = \mathbf{\Sigma}^{-1}$,  
\begin{align}
\rho_{ij | . } = \frac{ - k_{ij} }{ \sqrt{ k_{ii} }{ k_{jj}}}
\end{align}
Thus, the adjacency matrix is defined as $\mathbf{P} = \rho_{ij | . }$ for $i \neq j$ and  $\mathbf{P} = 0$, for $i = j$ which can be further decomposed as below
\begin{align}
\mathbf{P} = \mathbf{I} - \mathbf{D}^{-1/2}_{\mathbf{K}} \mathbf{K} \mathbf{D}^{-1/2}_{\mathbf{K}}, \ \text{where} \ \mathbf{D}_{\mathbf{K}} = \text{diag}[k_{11},...,k_{NN}] \equiv \mathbf{D}^{-1}_{\mathbf{\Sigma}}
\end{align}
Furthermore, the authors show that the system of linear regression equations which capture the financial dependencies in the network, it contains the proposed partial correlation adjacency matrix and corrects for the endogenous effect of variables used to capture these linear dependencies between the nodes of the financial network is given as
\begin{align}
X_i - m =  \mathbf{D}^{-1/2}_{\mathbf{K}} \mathbf{K} \mathbf{D}^{-1/2}_{\mathbf{K}} ( X_j - m) + \epsilon
\end{align}
   
The particular study verifies the network effects and the almost stylized facts by now in the literature of financial connectedness that the Australian financial institutions show a dramatic increase in all the various measures of centrality post the financial crisis of 2008. From a policy maker perspective the knowledge of such an indication during periods of financial turbulence would be considered as an early warning indication of potential systemic risk amplification and financial contagion. However, a drawback of this method, is that the network is not directional and thus it does not allow to capture any conditional tail dependencies similar to the methodology we are proposing. On the other hand, we are interested to examine the effects of network topology i.e., how the level of financial connectedness affects the asset sensitivity and the dynamic asset allocation decisions of investors. For this reason, even though the above literature provides compelling methodologies for constructing a network by decomposing for example the variance-covariance matrix into an adjacency matrix. Further applications include, \cite{huttner2016portfolio} who examine the implementation of portfolio allocation in graphs. \cite{davison1991exact} examine an exact conditional test based on partial correlations. The idea of recursively updating the eigenvalues of the covariance matrix are presented by \cite{yu1991recursive}. 

\end{small}

%%-------------------------------------------------------------------------%%
\newpage 

\section{Appendix B. Statistical Inference for entries of large precision matrices}

As a learning exercise to obtain useful asymptotic theory for the proposed risk matrix we follow the excellent derivations presented in the study of \cite{chang2018confidence}. In particular, consider the statistical inference for high-dimensional precision matrices. Specifically, we propose a data-driven procedure for constructing a class of simultaneous confidence regions for a subset of the entries of a large precision matrix. These confidence regions can be applied to test for specific structures of a precision matrix, and to recover its nonzero components. Let $\boldsymbol{y}_1,..., \boldsymbol{y}_n$ be $n$ observations from an $\mathbb{R}^p-$valued time series, where $\boldsymbol{y}_t = \big( y_{1,t},..., y_{p,t} \big)^{\prime}$ and each $\boldsymbol{y}_t$ has the constant first two moments, such that $\mathbb{E} ( \boldsymbol{y}_t ) = \boldsymbol{\mu}$ and $\mathsf{Cov} \left( \boldsymbol{y}_t \right) = \boldsymbol{\Sigma}$ for each $t$. Moreover, let $\boldsymbol{\Omega} = \boldsymbol{\Sigma}^{-1}$ be the precision matrix. We assume that $\left\{ \boldsymbol{y}_t \right\}$ is $\beta-$mixing in the sence that $\beta_k \to 0$ as $k \to \infty$, where
\begin{align}
\beta_k = \underset{ t }{ \mathsf{sup} } \ \mathbb{E} \left\{ \underset{ B \in \mathcal{F}_{t+k}^{ \infty }  }{ \mathsf{sup} } \left| \mathbb{P} \big( B | \mathcal{F}_{- \infty}^t \big) - \mathbb{P} (B) \right| \right\}.     
\end{align}

Specifically, we have that $\mathcal{F}_{- \infty}^t$ and $\mathcal{F}_{t+k}^{\infty}$ are the $\sigma-$fields generated respectively by $\left\{ \boldsymbol{y}_u \right\}_{ u \leq t }$ and $\left\{ \boldsymbol{y}_u \right\}_{ u \geq t + u }$. Notice that $\beta-$mixing is a mild condition for time series. It is known that causal ARMA processes with continuous innovation distributions, stationary Markov chains under some mild conditions and stationary GARCH models with finite second momements and continuous innovation distributions are all $\beta-$mixing.  

\subsection{Main Estimation Results}

Recall the relationship between a precision matrix and nodewise regressions. For a random vector $\boldsymbol{y} = \big( y_1,..., y_p \big)^{\prime}$ with mean $\boldsymbol{\mu} = \boldsymbol{0}$ and covariance $\boldsymbol{\Sigma}$, we consider the $p$ nodewise regressions such that
\begin{align}
y_{j_1} := \sum_{ j_2 \neq j_1 } \alpha_{ j_1, j_2 } y_{j_2} + \epsilon_{j_1}, \ \ \ j_1 = 1,...,p.    
\end{align}

Consider that $\boldsymbol{y}_{}- j_1 := \left\{ y_{j_2}: j_2 \neq j_1  \right\}$. Moreover, the regression error $\epsilon_{j_1}$ is uncorrelated with $\boldsymbol{y}_{-j_1}$ if and only if $\alpha_{ j_1, j_2 } = - \frac{ \omega_{j_1, j_2} }{  \omega_{j_1, j_1}  }$ for any $j_2 \neq j_1$. Under this condition it holds that, $\mathsf{Cov} \big( \epsilon_{j_1}, \epsilon_{j_2} \big) = \frac{ \omega_{ j_1, j_2 } }{ \omega_{ j_1,j_1} \omega_{ j_2,j_2} }$ for any $j_1$ and $j_2$. 

Moreover, let $\boldsymbol{\epsilon} = \big( \epsilon_1,..., \epsilon_p    \big)^{\prime}$ and $\boldsymbol{V} = \mathsf{Cov} ( \boldsymbol{\epsilon} ) = \left( v_{j_1, j_2} \right)_{ p \times p }$. Then, it holds that
\begin{align}
\boldsymbol{\Omega} = \left\{ \mathsf{diag} (\boldsymbol{V}) \right\}^{-1} \boldsymbol{V} \left\{ \mathsf{diag} (\boldsymbol{V}) \right\}^{-1}   
\end{align}
This relationship between $\boldsymbol{\Omega}$ and $\boldsymbol{V}$ provides a way to learn $\boldsymbol{\Omega}$ by the regression errors. In particular, since the error vector $\boldsymbol{\epsilon}$ is unobserved in practice, its "proxy" - the residuals of the node-wise regressions - can be employed to estimate the matrix $\boldsymbol{V}$. Let $\boldsymbol{\alpha}_j = \big( \alpha_{j,1},..., \alpha_{j,j-1}, - 1, \alpha_{j,j+1},..., \alpha_{j,p}  \big)^{\prime}$. For each $j = 1,..., p$, we may fit the high-dimensional linear regression given by 
\begin{align}
y_{j,t} = \sum_{ k \neq j } \alpha_{j,k} y_{k,t} + \varepsilon_{j,t}, \ \ \ t = 1,...,n.    
\end{align}
by Lasso shrinkage. Moreover, in the case that $\boldsymbol{\mu} \neq \boldsymbol{0}$, the regression specification is fitted using the centered data, $\boldsymbol{y}_t - \bar{\boldsymbol{y}}$, where $\bar{\boldsymbol{y}} = n^{-1} \sum_{t=1}^n \boldsymbol{y}_t$ is the sample mean.

%%-------------------------------------------------------------------------%%
\newpage 

Let $\widehat{\boldsymbol{\alpha}}_j$ be the Lasso estimator of $\boldsymbol{\alpha}_j$ defined as follows
\begin{align}
\widehat{\boldsymbol{\alpha}}_j = \underset{ \boldsymbol{\gamma} \in \Theta_j  }{ \mathsf{argmin} } \ \left[ \frac{1}{n} \sum_{t=1}^n \left( \boldsymbol{\gamma}^{\prime} \boldsymbol{y}_t \right)^2 + 2 \lambda_j | \boldsymbol{\gamma} |_1 \right],   
\end{align}
where $\Theta_j = \big\{ \boldsymbol{\gamma} = \left( \gamma_1,..., \gamma_p  \right)^{\prime} \in \mathbb{R}^p: \gamma_j = - 1 \big\}$ and $\lambda_j$ is the tuning parameter. Moreover, for each $t$, the residual $\widehat{\epsilon}_{j,t} = - \widehat{\boldsymbol{\alpha}}_j^{\prime} \boldsymbol{y}_t$ provides an estimate of $\epsilon_{j,t}$. 

Write with $\widehat{\boldsymbol{\epsilon}}_t := \big( \widehat{\epsilon}_{1,t},...,  \widehat{\epsilon}_{p,t} \big)^{\prime}$ and let $\widetilde{\boldsymbol{V}} = \left( \widetilde{v}_{j_1,j_2} \right)_{p \times p}$ be the sample covariance of $\left\{ \widehat{\boldsymbol{\epsilon}}_t \right\}_{t=1}^n$, where
\begin{align}
\widetilde{v}_{j_1, j_2} = n^{-1} \sum_{t=1}^n \varepsilon_{j_1,t} \varepsilon_{j_2,t}   
\end{align}
In particular, it is well-known that $n^{-1} \sum_{t=1}^n \epsilon_{j_1,t} \epsilon_{j_2,t}$ is an unbiased estimator of $v_{j_1,j_2}$, however replacing $\epsilon_{j_1,t}$ by $\widehat{\epsilon}_{j_1,t}$ will incur a bias term. Specifically, based on some sparsity conditions on $\boldsymbol{\Omega}$ and the growth rate of $p$ with respect to $n$, it holds that
\begin{align*}
\widetilde{v}_{j_1,j_2} - \frac{1}{n} \sum_{t=1}^n \epsilon_{j_1,t} \epsilon_{j_2,t} 
&= 
- \left( \widehat{\alpha}_{j_1,j_2} - \alpha_{j_1,j_2}   \right) \left( \frac{1}{n} \sum_{t=1}^n \epsilon_{j_2,t}^2 \right) \boldsymbol{1} \left\{ j_1 \neq j_2 \right\}  
\\
&\ \ \ - \left( \widehat{\alpha}_{j_2,j_1} - \alpha_{j_2,j_1} \right) \left( \frac{1}{n} \sum_{t=1}^n \epsilon_{j_1,t}^2 \right) \boldsymbol{1} \left\{ j_1 \neq j_2 \right\} + o_p \left(  \left\{ n \mathsf{log}(p) \right\}^{-1/2} \right). 
\end{align*}
Thus, to eliminate the bias, we employ an estimator for $v_{j_1,j_2}$ such that
\begin{align}
\widehat{v}_{j_1,j_2} :=
\begin{cases}
\displaystyle - \frac{1}{n} \sum_{t=1}^n \left( \widehat{\epsilon}_{j_1,t} \widehat{\epsilon}_{j_2,t} + \widehat{\alpha}_{j_1,j_2} \widehat{\epsilon}^2_{j_2,t} + \widehat{\alpha}_{j_2,j_1} \widehat{\epsilon}^2_{j_1,t} \right),  & \ j_1 \neq j_2  
\\
\displaystyle \frac{1}{n} \sum_{t=1}^n \widehat{\epsilon}_{j_1,t} \widehat{\epsilon}_{j_2,t},  & \ j_1 = j_2  
\end{cases}
\end{align}
Therefore, this implies that the elements of the inverse of the covariance matrix (precision matrix), $\omega_{j_1,j_2}$, can be estimated as below
\begin{align}
\widehat{\omega}_{j_1,j_2} = \frac{ \widehat{v}_{ j_1, j_2 } }{ \widehat{v}_{ j_1,j_1} \widehat{v}_{ j_2,j_2} }
\end{align}
for any $j_1$ and $j_2$. 

Notice that due to the well-behaved asymptotic theory above, similar limit results are valid in the case we consider other penalized methods such as the Dantzig estimation and the scaled Lasso. 

\medskip

\begin{condition}
The eigenvalues of the covariance matrix $\boldsymbol{\Sigma}$ are uniformly bounded away from zero and infinity.     
\end{condition}

\begin{condition}
There exist constants $K_3 > 0$ and $\gamma_3 > 0$ independent of $p$ and $n$ such that the $\beta-$mixing coefficient is bounded for any positive $k$, with $\beta_k \leq \mathsf{exp} \left( - K_3 k^{\gamma_3} \right)$.      
\end{condition}

%%-------------------------------------------------------------------------%%
\newpage 

\subsection{Confidence Regions}

Denote with $\boldsymbol{\Delta} = - n^{-1} \sum_{t=1}^n \left( \boldsymbol{\epsilon}_t \boldsymbol{\epsilon}_t^{\prime} - \boldsymbol{\Delta}   \right)$. Then, it follows that 
\begin{align}
\widehat{\boldsymbol{\Omega}} - \boldsymbol{\Omega} = \boldsymbol{\Pi} + \boldsymbol{ \mathcal{Y} }, \ \ \ \text{with} \ \ \boldsymbol{\Pi} = \left\{ \mathsf{diag}(\boldsymbol{V})\right\}^{-1} \boldsymbol{\Delta} \left\{ \mathsf{diag}(\boldsymbol{V})\right\}^{-1}.  
\end{align}

\begin{condition}
There exists constant $K_4 > 0$ such that
\begin{align}
\underset{ b \to \infty }{ \mathsf{lim inf} } \ \underset{ 1 \leq \ell \leq n + 1 - b }{ \mathsf{inf} } \ \mathbb{E} \left( \left| \frac{1}{ b^{1/2} } \sum_{t= \ell}^{ \ell + b - 1 }  \eta_{j,t} \right| \right) > K_4,   
\end{align}
for each $j = 1,..., r$.
\end{condition}

\medskip

\begin{remark}
Notice that the Condition above is necessary for the validity of the Gaussian approximation for dependent data. Moreover, Davydov inequality entails that 
\begin{align}
\underset{ b \to \infty }{ \mathsf{lim sup} } \ \underset{ 1 \leq \ell \leq n + 1 - b }{ \mathsf{sup} } \ \mathbb{E} \left( \left| \frac{1}{ b^{1/2} } \sum_{t= \ell}^{ \ell + b - 1 }  \eta_{j,t} \right|^2 \right) <  K_5,    
\end{align}
for some universal constant $K_5 > 0$. These two conditions match the requirements of Gaussian approximation imposed on the long-run covariance of $\left\{ \eta_{j,t} \right\}_{ t = \ell }^{ \ell + b - 1 }$, for $j = 1,...,r$ and $\ell = 1,..., n+1 -b$.

Moreover, notice that if $\left\{ \eta_{j,t} \right\}$ is stationary, then it holds that 
\begin{align}
\mathbb{E} \left( \left| \frac{1}{ b^{1/2} } \sum_{t= \ell}^{ \ell + b - 1 }  \eta_{j,t} \right|^2 \right)  = \mathbb{E} \left( \eta_{j,1}^2 \right) + \sum_{k=1}^{b-1} \left(  1 - \frac{k}{b} \right) \mathsf{Cov} \big( \eta_{j,1}, \eta_{j,k+1} \big)
\end{align}
Therefore, under the stationarity assumption on each sequence $\left\{ \eta_{j,t}  \right\}$, the condition above is equivalent to $\sum_{k=0}^{ \infty } \mathsf{Cov} \big( \eta_{j,1}, \eta_{j,k+1} \big) > K_4$ for any $j \in \left\{ 1,..., r \right\}$. 
\end{remark}

\medskip

\begin{remark}
When $p$ grows with $n$, the concept of convergence in distribution does not apply, and thus different tools should be used to derive an appropriate critical value for the test. Furthermore, the distribution of $\mathsf{max}_{ 1 \leq j \leq p } Y_j$ is typically unknown because the covariance structure of $Y$ is unknown.     
\end{remark}

\medskip

\begin{remark}
Notice that within our setting the convergence rate could be slower due to the estimation of the generated covariate. Overall the design matrix is $\boldsymbol{X}_n$ is stochastic and dependent. In addition we might have less restriction on the moment assumption of the noise sequence compare to the strong sub-Gaussian assumption for instance in the change-in-mean of \textit{i.i.d} data discussed in the literature. Consequently, the asymptotics is more subtle and ideas from mixing have to be adopted to establish the tail probability bound. As a result the dependence and the relaxed moment assumption, the convergence rate of the tail-estimates driven covariance matrix is slower than in the case in which the covariance matrix is based on the error estimates of nodewise regressions.        
\end{remark}

%%-------------------------------------------------------------------------%%
\newpage

More specifically, for each $r-$dimensional vector whose $j-$th element is
\begin{align}
\frac{ \epsilon_{ \chi_1(j),t } \epsilon_{ \chi_2(j),t } - \nu_{ \boldsymbol{\chi}(j) } }{  \nu_{ \chi_1(j), \chi_1(j) }  \nu_{ \chi_2(j), \chi_2(j) }  }    
\end{align}
where $\boldsymbol{\chi} (.) = \big\{ \chi_1(.), \chi_2(.) \big\}$ is a bijective mapping from $\left\{ 1,..., r \right\}$ to $\mathcal{S}$ such that $\boldsymbol{\Omega}_{\mathcal{S}} = \left\{  \omega_{ \boldsymbol{\chi}(1) },...,  \omega_{ \boldsymbol{\chi}(r) }  \right\}$. Then, we have that 
\begin{align}
\boldsymbol{\Pi}_{ \mathcal{S} } = - \frac{1}{n} \sum_{t=1}^n \widetilde{\boldsymbol{\sigma}}_t  
\end{align}
Moreover, the long-run covariance matrix of $\left\{  \widetilde{\boldsymbol{\sigma}}_t \right\}_{t=1}^n$ is given by 
\begin{align}
\boldsymbol{W} = \mathbb{E} \left[ \left(  \frac{1}{ \sqrt{n} } \sum_{t=1}^n \widetilde{\boldsymbol{\sigma}}_t \right) \left(  \frac{1}{ \sqrt{n} } \sum_{t=1}^n \widetilde{\boldsymbol{\sigma}}_t \right)^{\prime} \right].    
\end{align}
Let $\boldsymbol{\eta}_t = \big( \eta_{1,t},..., \eta_{r,t}  \big)^{\prime}$, where $\eta_{j,t} = \epsilon_{\chi_1(j),t} \epsilon_{\chi_2(j),t} - \nu_{\boldsymbol{\chi}(j)}$. Then, the matrix $\boldsymbol{W}$ can be written as 
\begin{align}
\boldsymbol{W} = \boldsymbol{H} \mathbb{E} \left[ \left(  \frac{1}{ \sqrt{n} } \sum_{t=1}^n \boldsymbol{\eta}_t \right) \left(  \frac{1}{ \sqrt{n} } \sum_{t=1}^n \boldsymbol{\eta}_t \right)^{\prime} \right] \boldsymbol{H}
\end{align}
where 
\begin{align}
\boldsymbol{H} = \mathsf{diag} \left\{  \nu^{-1}_{ \chi_1(1), \chi_1(1) } \nu^{-1}_{ \chi_2(1), \chi_2(1) },..., \nu^{-1}_{ \chi_1(r), \chi_1(r) } \nu^{-1}_{ \chi_2(r), \chi_2(r) }   \right\}    
\end{align}

To study the asymptotical distribution of the average of the temporally dependent sequence $\left\{ \widetilde{\boldsymbol{\sigma}}_t \right\}_{t=1}^n$ and its long-run covariance $\boldsymbol{W}$, we introduce the following condition on $\left\{ \boldsymbol{\eta} \right\}_{t=1}^n$. 

\medskip

Suppose that we are interested in approximating the distribution of $n^{1/2} \left| \widehat{\boldsymbol{\Omega}}_{\mathcal{S}} - \boldsymbol{\Omega}_{\mathcal{S}} \right|_{\infty}$. 

\medskip

\begin{theorem}
Let $\boldsymbol{\xi} \sim \mathcal{N} \big( \boldsymbol{0}, \boldsymbol{W}  \big)$. Under the previous conditions we have that 
\begin{align}
\underset{ x > 0 }{ \mathsf{sup} } \left|  \mathbb{P} \left( n^{1/2} \left| \widehat{\boldsymbol{\Omega}}_{\mathcal{S}} - \boldsymbol{\Omega}_{\mathcal{S}} \right|_{\infty} > x \right) - \mathbb{P} \big( \left| \boldsymbol{\xi} \right|_{\infty} > x \big) \right| \to 0, \ \ \text{as} \ \ n \to \infty.    
\end{align}
In other words, the Kolmogorov distance between the distributions of $n^{1/2} \left| \widehat{\boldsymbol{\Omega}}_{\mathcal{S}} - \boldsymbol{\Omega}_{\mathcal{S}} \right|_{\infty}$ and $\left| \boldsymbol{\xi} \right|_{\infty}$ converges to zero. On the other hand, notice that $\left| \widehat{\boldsymbol{\Omega}}_{\mathcal{S}} - \boldsymbol{\Omega}_{\mathcal{S}} \right|$ may converge weakly to an extreme value distribution, which will require some further assumptions on the structure of $\boldsymbol{W}$. 
\end{theorem}
Define an estimator of the matrix $\boldsymbol{\Gamma}_k$ such that
\begin{align}
\widehat{\boldsymbol{\Gamma}}_k = 
\begin{cases} 
\displaystyle \frac{1}{n} \sum_{t=k+1}^n \widehat{\boldsymbol{\eta}}_{t}  \widehat{\boldsymbol{\eta}}^{\prime}_{t-k}, k \geq 0  
\\
\displaystyle \frac{1}{n} \sum_{t=-k+1}^n  \widehat{\boldsymbol{\eta}}_{t+k}  \widehat{\boldsymbol{\eta}}^{\prime}_{t}, k < 0
\end{cases}
\end{align}

%%-------------------------------------------------------------------------%%
\newpage 

\section{Appendix C. Asymptotic theory for the VaR-$\Delta$CoVaR risk matrix}

\subsection{Stationary Case}

The asymptotic normality of the VaR estimator, under stationarity assumption, is derived by \cite{he2020inference} in their Theorem 1, which is provided below. Let $\left\{ D_{nt} \right\}$ be a martingale difference array adapted to the filtration 
\begin{align}
\left\{ \mathcal{F}_t \right\}_{ t \geq 0 } = \left\{ \sigma \left(\epsilon_1,..., \epsilon_t, \mathbf{Z}_1,..., \mathbf{Z}_t, \mathbf{Z}_{t+1} \right) \right\}, \ \text{where} \ \mathbf{Z}_1 = \left( 1, \mathbf{X}_1^{\prime} \right)^{\prime} 
\end{align}
such that $\mathbf{X}_1$ is a stationary regressor and $\Omega = \mathbb{E} \left( \mathbf{Z}_1 \mathbf{Z}_1^{\prime} \right)$. 

\begin{theorem}
Under conditions A1-A4 in the paper of He et al (AoS, 2021), for $\alpha \in (0,1)$, 
\begin{align}
\label{normal.expression}
\sqrt{n} \left\{ \widehat{ \text{VaR}}_{\mathbf{x}} \left( \alpha \right) -  \text{VaR}_{ \mathbf{x} } \left( \alpha \right)  \right\}
\overset{ d }{ \to } \mathcal{N} \left( 0, \omega^2 + \sigma^2 \mathbf{z}^{\prime} \Omega^{-1} \mathbf{z} + \Delta \right)
\end{align}
where $\Delta$ and $\omega^2$ as defined in the paper.  
\end{theorem}

\medskip

\begin{remark}
Notice that the "linear predictive regression" defined in the particular paper, considers the pair  $\left\{ Y_t, X_t \right\}$ at the same lag, so in our setting we aim to adapt the quantile predictive regression specification as in Lee (2016) e.g., to study stationary versus non-stationary case. Therefore, in that case we will need to adapt the particular limiting distribution given by expression \eqref{normal.expression}.  

Then, since the diagonal of the risk matrix includes random variables  $\text{VaR}^{+}_{1},...,\text{VaR}^{+}_{1}$ these elements should follow the same distribution as the corresponding limiting distribution for the case of stationarity (i.e., no LUR specification) or the mixed Gaussian in the case of LUR specification and a possible use of the instrumentation method to control for the abstract degree of persistence. 

The next step, would be to determine the limiting distribution of the elements $\text{CoVaR}_{(i,j)}$ under both stationarity and nonstationarity which will give us to different asymptotic results. More specifically, in the case we use the quantile predictive regression as defined by \cite{lee2016predictive} then, we will need to consider the asymptotic distribution for the IVX estimator under the assumption of generated regressor (i.e., the VaR in the first step specification). All elements in the matrix which correspond to $\text{CoVaR}_{(i,j)}$ should be identically distributed. A possible tail dependence might appear because of employing common predictors when estimating the elements of the risk matrix that correspond to CoVaR for different $i$ and $j$. 

Assuming all above steps are analytical tractable and we obtain the limiting distribution for the different cases, then the challenging step would be to determine the limiting distribution of the term 
\begin{align}
(\text{VaR}^{+}_{i} \text{VaR}^{+}_{j})^{1/2} \Delta \text{CoVaR}_{(i,j)}
\end{align} 
and obviously of the whole risk matrix $\mathbf{\Gamma}$. \end{remark}

%%-------------------------------------------------------------------------%%
\newpage 

Determining the limiting distribution of the proposed risk matrix, will also allow us to do statistical inference on the estimated parameters.  

We begin by examining the limiting distribution of the term 
\begin{align}
\sqrt{n} \left\{ \widehat{ \text{VaR}}_{\mathbf{x}} \left( \alpha \right) - \text{VaR}_{ \mathbf{x} } \left( \alpha \right)  \right\}
\overset{ d }{ \to } \mathcal{N}_1 
\end{align}
under the assumption that we have a quantile predictive regression model with stationary predictors (i.e., no LUR specification). For instance, by defining a suitable martingale difference sequence we can use Theorem 3.2 from Hall and Hyde (1980) to establish the asymptotic normality of our VaR estimator.

Let 
\begin{align}
D_{nt} = \frac{1}{ \sqrt{n} } \left\{ - \frac{ I \left( \epsilon_t \leq \gamma \right) - \alpha }{ F_{\epsilon}^{\prime} \left( \gamma \right)} + \epsilon_t \mathbf{Z}_T^{\top} \Omega^{-1} \mathbf{z} - \epsilon_t \mathbf{Z}_T^{\top} \Omega^{-1} \mathbb{E} \left( \mathbf{Z}_1 \right) \right\}.
\end{align}
For Theorem 1 above we have that $\mathbf{z} = \left( 1, \mathbf{x}^{\top} \right)^{\top}$, $\omega^2 = \frac{ \alpha ( 1 - \alpha )  }{ \left\{ F_{\epsilon}^{\prime} \left( F_{\epsilon}^{-1} \left( \alpha \right) \right) \right\}^2 }$, and $\Delta = \Delta_1 + \Delta_2$ with 
\begin{align}
\Delta_1 &= \sigma^2 \mathbb{E} \left( \mathbf{Z}_1^{\top} \right) \Omega^{-1}  \mathbb{E} \left( \mathbf{Z}_1 \right) + 2 \frac{ \mathbb{E} \left( \epsilon_1 I \left( \epsilon_1 \leq F_{\epsilon}^{-1} \left( \alpha \right) \right) \right) }{   F_{\epsilon}^{\prime} \left( F_{\epsilon}^{-1} \left( \alpha \right) \right)} \mathbb{E} \left( \mathbf{Z}_1^{\top} \right)  \Omega^{-1} \mathbb{E} \left( \mathbf{Z}_1 \right)
\\
\nonumber
\\
\Delta_2 &= - 2 \sigma^2 \mathbb{E} \left( \mathbf{Z}_1^{\top} \right) \Omega^{-1} \mathbf{z} - 2 \frac{ \mathbb{E} \left( \epsilon_1 I \left( \epsilon_1 \leq F_{\epsilon}^{-1} \left( \alpha \right) \right) \right) }{   F_{\epsilon}^{\prime} \left( F_{\epsilon}^{-1} \left( \alpha \right) \right)} \mathbb{E} \left( \mathbf{Z}_1^{\top} \right)  \Omega^{-1} \mathbf{z}. 
\end{align} 

\begin{theorem}
Under conditions B1-B6, for $\alpha \in (0,1)$, 
\begin{align}
\sqrt{n} \left\{ \widehat{ \text{VaR}}_{\mathbf{x}} \left( \alpha \right) - \text{VaR}_{ \mathbf{x} } \left( \alpha \right)  \right\}
\overset{ d }{ \to } \mathcal{N} \left( 0, \omega^2 + \sigma^2 \widetilde{z}^{\top} \widetilde{\Omega}^{-1} \widetilde{z} + \widetilde{\Delta} \right),  
\end{align}
where $\widetilde{z} = \left( 1, x_1,..., x_{k-1} \right)^{\top}$, $\omega^2 = \frac{ \alpha ( 1 - \alpha )  }{ \left\{ F_{\epsilon}^{\prime} \left( F_{\epsilon}^{-1} \left( \alpha \right) \right) \right\}^2 }$, and $\widetilde{ \Delta } = \widetilde{ \Delta }_1 + \widetilde{ \Delta }_2$ with 
\begin{align*}
\widetilde{ \Delta}_1 &= \sigma^2 \mathbb{E} \left( \widetilde{ \mathbf{Z}}_1^{\top} \right) \widetilde{ \Omega }^{-1}  \mathbb{E} \left( \widetilde{ \mathbf{Z}}_1 \right) + 2 \frac{ \mathbb{E} \left( \epsilon_1 I \left( \epsilon_1 \leq F_{\epsilon}^{-1} \left( \alpha \right) \right) \right) }{   F_{\epsilon}^{\prime} \left( F_{\epsilon}^{-1} \left( \alpha \right) \right)} \mathbb{E} \left( \widetilde{ \mathbf{Z}}_1^{\top} \right)  \Omega^{-1} \mathbb{E} \left( \widetilde{ \mathbf{Z}}_1 \right)
\\
\nonumber
\\
\widetilde{ \Delta}_2 &= - 2 \sigma^2 \mathbb{E} \left( \widetilde{ \mathbf{Z}}_1^{\top} \right) \widetilde{ \Omega }^{-1} \widetilde{ \mathbf{z}} - 2 \frac{ \mathbb{E} \left( \epsilon_1 I \left( \epsilon_1 \leq F_{\epsilon}^{-1} \left( \alpha \right) \right) \right) }{   F_{\epsilon}^{\prime} \left( F_{\epsilon}^{-1} \left( \alpha \right) \right)} \mathbb{E} \left( \widetilde{ \mathbf{Z}}_1^{\top} \right) \widetilde{ \Omega }^{-1} \widetilde{ \mathbf{z} }. 
\end{align*} 
\end{theorem}

%%-------------------------------------------------------------------------%%
\newpage 

\subsection{Nonstationary Case: (LUR specification)}

Let $\left( \mathbb{P}, \mathcal{F}, \Omega \right)$ be a probability space with $\mathcal{F}$ the information set and an equipped $\sigma-$algebra. For the nonstationary case, we consider that the predictors in the model are modelled by the Local-to-Unit root specification. To begin with, we consider the estimation procedure for the CoVaR risk measure, which is based on the two-step quantile predictive regression specification. Moreover, the estimation procedure takes into account the graph representation with nodes being the stocks.  

Below we provide a discussion in the case one would consider the estimation of a large covariance matrix of a multivariate nearly-integeated process with unknown degree of persistence. Consider that the $p-$dimensional regressors $\boldsymbol{X}_t$ follows a VAR(1) process with a LUR autoregression matrix as
\begin{align}
\boldsymbol{X}_t = \left( \boldsymbol{I}_p - \frac{ \boldsymbol{C}_p }{ n^{\alpha} } \right) \boldsymbol{X}_{t-1} + \boldsymbol{u}_t, \ \ \ t = 1,...,n 
\end{align}

The above equation can be also written as below
\begin{align}
\boldsymbol{R}_n ( c, \alpha )  \boldsymbol{X} = \boldsymbol{U}
\end{align}
where the non-random $n \times n$ matrix $\boldsymbol{R}_n ( c, \alpha )$ takes the bi-diagonal form as below
\begin{align}
\boldsymbol{R}_n ( c, \alpha ) = 
\begin{bmatrix}
\textcolor{red}{1} & 0 & 0 & \hdots & 0  & 0 
\\
\textcolor{blue}{- \rho_n (c, \alpha )} & \textcolor{red}{1} & 0 & \hdots & 0  & 0 
\\
0 & \textcolor{blue}{- \rho_n (c, \alpha )} & \textcolor{red}{1} & \vdots & 0  & 0 
\\
\vdots & \vdots & \vdots & \vdots & \vdots  & \vdots
\\
0 & 0 & 0 & 0 & \textcolor{blue}{- \rho_n (c, \alpha )}  & \textcolor{red}{1} 
\end{bmatrix}
\end{align}
with $\textcolor{blue}{- \rho_n (c, \alpha ) = 1 - c / n^{\alpha}} $.

Moreover, let $\boldsymbol{H}_n = \boldsymbol{I}_n - \frac{1}{n} \mathbf{1} \mathbf{1}^{\prime}$ denote the $n \times n$ centering matrix and $\boldsymbol{X}_H = \boldsymbol{H}_n \boldsymbol{X}$ the corresponding centred data matrix and therefore the corresponding centered covariance matrix can be written in the following form 
\begin{align}
\boldsymbol{S}_n = \frac{1}{p} \boldsymbol{X}_H \boldsymbol{X}_H^{\prime}. 
\end{align}

Notice also that we have that $\boldsymbol{X} = \boldsymbol{R}_n ( c, \alpha )^{-1} \boldsymbol{U}$ which implies that the covariance matrix can be written in the following form
\begin{align}
\boldsymbol{S}_n = \frac{1}{p} \left( \boldsymbol{H}_n \boldsymbol{R}_n ( c, \alpha )^{-1} \boldsymbol{U} \right) \left( \boldsymbol{H}_n \boldsymbol{R}_n ( c, \alpha )^{-1} \boldsymbol{U} \right)^{\prime}  
\end{align}
which provides a way of expressing the covariance matrix with respect to persistence, under the assumption that the data are generated as an VAR(1) model with a LUR coefficient matrix.

Furthermore, the uncentered matrix can be written as below 
\begin{align}
\boldsymbol{S}_n^{uc} = \frac{1}{p} \boldsymbol{X} \boldsymbol{X}^{\prime}.
\end{align}
and using the VAR(1) process we obtain
\begin{align}
\boldsymbol{S}_n^{uc} = \left( \boldsymbol{R}_n ( c, \alpha )^{-1} \boldsymbol{U} \right) \left( \boldsymbol{R}_n ( c, \alpha )^{-1} \boldsymbol{U} \right)^{\prime}
\end{align}

Denote with $\hat{ \lambda}_k$ and $\hat{ \lambda}^{uc}_k$ the $k-$th largest eigenvalues of $\boldsymbol{S}_n$ and $\boldsymbol{S}_n^{uc}$ respectively. Therefore, we are interested to examine the asymptotic behaviour of the eigenvalues of the matrices $\boldsymbol{S}_n$ and $\boldsymbol{S}_n^{uc}$ in a high-dimensional setting.

\bigskip

Notice that the above $p-$regressors are employed as lagged regressors in the various quantile predictive regression models. Furthermore, as additional covariate we consider the systemic risk covariate. Moreover, in the case when each of the individual equations is allowed to have different regressors then this implies the existence of different persistence properties across the node specific equations. However, we do not examine the particular scenario in our setting.

%%-------------------------------------------------------------------------%%
%\newpage 
   
\bibliographystyle{apalike}

{\small
\bibliography{myreferences1}}

%%-------------------------------------------------------------------------%%
\newpage

\centerline{ \large \textbf{SUPPLEMENTARY APPENDIX} }
\
\vspace*{-0.3 em}
\centerline{ \textbf{Statistical Estimation for Covariance Structures with Tail Estimates} }
\

\vspace*{-0.7 em}
\centerline{ \textbf{using Nodewise Quantile Predictive Regression Models} }

\paragraph{I. Discussion of Computational Aspects}

\

In this Section, we discuss some relevant computational aspects to our framework. In terms of the implementation of the algorithm we use the Statistical Package R. To reduce the execution time we utilize parallel programming techniques as well as related R packages for this purpose, such as the \texttt{Rcpp} library and the \texttt{Rccp Armadillo} linear algebra library. For example, we  use parallel computing techniques with R using the Iridis4 High Performance Computing Facility of the University of Southampton. An innovation of our proposed methodology which reduces significantly computational time is the algorithm we use for the estimation of the risk matrix. The R installation on the Iridis4 cluster allows the use of 16 cores per submitted job. Moreover, by submitting a number of jobs at the same time we were able to utilize in the order of 64 cores in parallel for our computation. 
We define with $\boldsymbol{\Gamma}$ the $\mathsf{VaR-\Delta CoVaR}$ risk matrix. Then, the elements of the risk matrix are estimated using the procedure below\footnote{Notice that the novelty of our Algorithm 2 is the parallel estimation of the elements of the matrix by avoiding restimating the VaR in each iteration we estimate the corresponding CoVaR. Researchers investigating estimation problems for high dimensional covariance and precision matrices have proposed various methodologies for estimating $\mathbf{\Sigma}^{-1}$ column-by-column implementable with parallel computing such as the nodewise Lasso and the constrained $\ell_1-$minimization for inverse matrix estimation (CLIME) (\cite{shu2019estimationmatrix}).}. Notice that we correct for possible estimation bias due to the existence of the generated regressor using a bootstrap procedure which is implemented internally in R when estimating each pair of $\mathsf{VaR}$ and $\mathsf{CoVaR}$'s using the R package \texttt{quantreg}.

\begin{verbatim}

VaR_DCoVar_forecast_function <- function( Nr_C = Nr_C, nhist = nhist, 
                          returns = returns_hist, macro = macro_hist, tau = tau )
{#begin of function
  
  # Initialize inputs
  Nr_C  <- Nr_C
  nhist <- nhist
  tau   <- tau
 
  returns <- returns_hist
  macro   <- macro_hist
  
  nr <- NROW(returns)
  p <- ncol(macro)
  
  # Step 1: Estimate the VaR-DCoVaR matrix (no change of signs in this step)
  VaR.DCoVar.forecast <- Risk_Matrix_forecast_function( Nr_C = Nr_C, nhist = nhist, returns=returns_hist, macro=macro_hist, tau = tau )
  
  # Step 2: Take the negative VaRs and DCoVaRs
  VaR.DCoVar.forecast.positive <- VaR.DCoVar.forecast
  for (i in 1: Nr_C )
  {
    for (j in 1: Nr_C )
    {  
      VaR.DCoVar.forecast.positive[i,j] <- - ( VaR.DCoVar.forecast[i,j] )
    }
  }
  
  # Second I construct the non-symmetric version of the proposed Gamma risk matrix
  # Take the estimated VaRs into a vector
  VaR.forecasts <- matrix( 0, nrow = Nr_C, ncol = 1 )
  for (i in 1:Nr_C )
  {
    VaR.forecasts[i,1] <- abs( VaR.DCoVar.forecast.positive[i,i] )
  }
  VaR.DCoVar.forecast.positive.new <- diag( Nr_C )
  
  # Step 3: Construct our proposed risk matrix 
  for (i in 1:Nr_C)
  {
    for (j in 1:Nr_C)
    {# begin for-loop
      if (j!=i)
      {# begin if 
        VaR.DCoVar.forecast.positive.new[j,i] <- ( ( VaR.forecasts[j,1]*VaR.forecasts[i,1] )^0.5 ) * VaR.DCoVar.forecast.positive[j,i]
      }# end if
    }# end for-loop
  }
  
  # Next I replace the diagonal with the positive estimated VaRs
  for (i in 1:Nr_C )
  {
    VaR.DCoVar.forecast.positive.new[i,i] <- VaR.forecasts[i,1]
  }
  
  # Step 4: Take the symmetric part of the Gamma matrix
  risk.matrix <- 0.5 * ( VaR.DCoVar.forecast.positive.new + t(VaR.DCoVar.forecast.positive.new) )
  
  return( risk.matrix )
  
}#end of function


\end{verbatim}

%%-------------------------------------------------------------------------%%
\newpage 

\paragraph{II. Discussion of Asymptotic Behaviour of Eigenvalues of Large Covariance Matrices}

\

Related studies include \cite{fan2017tracy} who consider estimating the unknown covariance structure with linear random effect models and construct an asymptotic test of the global sphericity null hypothesis, such that $H_0: \Sigma_r = \sigma_t^2$, for every $r \in \left\{ 1,..., k \right\}$.  

\begin{lemma}
Assume that assumptions holds. Then, there exists $n_0 \geq 1$ such that, for all $n > n_0$, $A_n$ is diagonalizable in the form 
\begin{align}
A_n = P_n D_n P_n^{-1}, \ \ \text{with} \ \ \ D_n = \mathsf{diag} \big( \lambda_{n,1},..., \lambda_{n,p} \big)    
\end{align}
containing order distinct eigenvalues $1 > | \lambda_{n,1} | \geq ... \geq | \lambda_{n,p} | > 0$. In addition, it holds that $\norm{ P_n } \leq \mathcal{C}$ and $\norm{ P_n^{-1} } \leq \mathcal{C}$.
\end{lemma}

\medskip

\begin{theorem}
Assume that conditions hold and that $\mathbb{E} \big[ \left|  \epsilon_1 \right|^{ 2 + \nu } \big] = \eta_{\nu} < + \infty$ for some $\nu > 0$. Then, if the eigenvalues of $A$ are real, we have the asymptotic normality such that 
\begin{align}
\sqrt{n} V_n^{-1/2} P_n^{\top} \left(  \widehat{\theta}_n - \theta_n \right)  \overset{d}{\to}  \mathcal{N} \big( 0, H_0^{-1} \big)
\end{align}
where $\widehat{\theta}_n$ is the OLS estimator in the nearly unstable $AR(p)$ process, $V_n$ corresponds to the matrix of rates and $H_0$ is a standardized positive precision matrix. Furthermore, if some eigenvalues of $A$ are complex, we have the asymptotic normality such that 
\begin{align}
\sqrt{n \kappa_n } \langle L_n,  \widehat{\theta}_n - \theta_n \rangle \overset{d}{\to} \mathcal{N} \big( 0, h_0^2 \big),
\end{align}
where
\begin{align}
L_n = \left( 1, \frac{1}{ \lambda_1 \rho_n },...,  \frac{1}{ ( \lambda_1 \rho_n )^{p-1} } \right)^{\top} \ \ \ \text{and} \ \ \ h_0^2 = \frac{ 2c }{ \pi^2_{11} } > 0. 
\end{align}
\end{theorem}

\begin{proof}
Suppose that the eigenvalues of $A$ are real, and consider the filtration $\mathcal{F}_{n,k} = \sigma \big( \Phi_{n,0}, \varepsilon_1,..., \varepsilon_k \big)$, $\forall n \geq 1, \forall 1 \leq k \leq n$. 

For all $a \in \mathbb{R}^p \backslash \left\{ 0 \right\}$, let also 
\begin{align}
m_{n,k}^{(a)} = a^{\top} V_n^{1/2} P_n^{-1} \Phi_{n,k-1} \varepsilon_k \ \ \ \text{and} \ \ \ M_n^{(a)} = \sum_{k=1}^n m_{n,k}^{(a)}   
\end{align}
Notice that the sequence $\left( m_{n,k}^{(a)} \right)$ is a scalar martingale difference array with respect to $\mathcal{F}_{n,k}$ at fixed $n$ and for $1 \leq k \leq n$. Moreover, the predictable quadratic variation of $M_n^{ (a) }$ is given by 
\begin{align*}
\langle M^{ (a) } \rangle_n &= \sum_{k=1}^n \mathbb{E} \left[ \left( m_{n,k}^{(a)} \right)^2 \big| \mathcal{F}_{n,k-1}  \right] 
\\
&=
\sigma^2 a^{\top} V_n^{1/2} P_n^{-1} \left[ \sum_{k=1}^n \Phi_{n,k-1} \Phi_{n,k-1}^{\top} \right] \left( P_n^{-1} \right)^{\top} V_n^{1/2} a
\end{align*}
since $\left( \varepsilon_k \right)$ is a white noise and $\Phi_{n,k-1}$ is $\mathcal{F}_{n,k-1}-$measurable, which results to the following convergence
\begin{align}
\frac{1}{n} \langle M^{ (a) } \rangle_n \overset{p}{\to} \sigma^2 a^{\top} H \alpha > 0.   
\end{align}
where $H$ is the covariance matrix which is positive definite. 

To apply the central limit theorem for arrays of martingales, it remains to show that the Linderberg's condition is satisfied such that 
\begin{align}
\forall \ \varepsilon > 0, \ \frac{1}{n} \sum_{k=1}^n \mathbb{E} \left[ \left( m_{n,k}^{(a)} \right)^2 \boldsymbol{1} \left\{  \left| m_{n,k}^{(a)} \right| > \epsilon \sqrt{n} \right\} \big| \mathcal{F}_{n,k-1}  \right] \overset{p}{\to} 0.
\end{align}
In order to prove the above result we first see that for any $1 \leq k \leq n$,
\begin{align*}
\mathbb{E} \left[ \left( m_{n,k}^{(a)} \right)^2 \boldsymbol{1} \left\{  \left| m_{n,k}^{(a)} \right| > \epsilon \sqrt{n} \right\} \big| \mathcal{F}_{n,k-1}  \right] &= T_{n,k-1}^{ (a) } \xi_{n,k}
\\
&\leq T_{n,k-1}^{ (a) } \underset{ 1 \leq k \leq n }{ \mathsf{sup} }  \xi_{n,k} 
\end{align*}
where $T_{n,k-1}^{ (a) } \overset{\Delta}{=} a^{\top} V_n^{1/2} P_n^{-1} \Phi_{n,k-1} \Phi_{n,k-1}^{\top} \left( P_n^{-1} \right)^{\top} V_n^{1/2} a > 0$ and using the Holder's and Markov's inequalities we get that
\begin{align*}
\xi_{n,k} &\overset{\Delta}{=} 
\mathbb{E} \left[ \varepsilon_k^2 \boldsymbol{1} \left\{  \left| m_{n,k}^{(a)} \right| > \epsilon \sqrt{n} \right\} \big| \mathcal{F}_{n,k-1}  \right] 
\\
&\leq \eta_{\nu}^{ \frac{2}{ (2 + \nu) } } \mathbb{P} \left( \left(  m_{n,k}^{(a)} \right)^2 > \epsilon^2 n \big| \mathcal{F}_{n,k-1} \right)^{ \frac{\nu}{ ( 2 + \nu) } } \leq \mathcal{C} \left(  \frac{ T_{n,k-1}^{(a)}  }{n} \right)^{ \frac{2}{ (2 + \nu) } }
\end{align*}
Therefore, it is sufficient to establish that 
\begin{align}
\forall \ a \in \mathbb{R}^p \backslash \left\{ 0 \right\}, \ \ \ \frac{ M_n^{(a)} }{ \sqrt{n} } \overset{d}{\to} \mathcal{N} \big( 0, \sigma^2 a^{\top} H a \big).     
\end{align}
Since $a$ is arbitrary, we can use the Cramer-Wold device to get the convergence of the $p-$dimensional vector such that 
\begin{align}
 \frac{ M_n^{(a)} }{ \sqrt{n} } \overset{d}{\to} \mathcal{N} \big( 0, \sigma^2 H \big).         
\end{align}
with 
\begin{align}
M_n = V_n^{1/2} P_n^{-1} \sum_{k=1}^n \Phi_{n,k-1} \varepsilon_k.     
\end{align}
   
\end{proof}

\begin{remark}
Notice that the above results are given in the excellent study of \cite{badreau2023consistency} and are particularly useful for obtaining relevant asymptotic theory for the eigenvalues of the proposed risk matrix in this article. Additional useful properties to be considered in this direction include the separation property of tail or extreme events. In particular the separation property shows that extremes induced from different factors belong to independent blocks. Further aspects relevant to the separation property are more formally in the study of \textcolor{blue}{Krizmani\'c and Katsouris (2023+)}: "\textit{Weak Convergence of Self-Normalized Partial Sum Processes in the $M_1$ topology}".  
\end{remark}

\end{document}